\documentclass[review]{elsarticle}

\usepackage{lineno,hyperref}
\usepackage{graphicx}
\usepackage{balance}

\usepackage{amsmath}
\usepackage{amsthm}
\usepackage{amsfonts}
\usepackage{algorithm}
\usepackage{algorithmicx}
\usepackage{algpseudocode}
\usepackage{subfigure}
\usepackage{color}
\usepackage{epstopdf}
\usepackage{url}

\usepackage{wrapfig,psfig,epsf}
\usepackage{bbding}

\newtheorem{theorem}{Theorem}
\newtheorem{definition}{Definition}
\newtheorem{example}{Example}[section]
\modulolinenumbers[5]

\journal{Journal of \LaTeX\ Templates}

\begin{document}

\begin{frontmatter}

\title{LAQP: Learning-based Approximate Query Processing}


\author[M. Zhang et al]{Meifan Zhang, Hongzhi Wang\Envelope}
\address{Department of Computer Science and Technology, Harbin Institute of Technology}

\begin{abstract}
Querying on big data is a challenging task due to the rapid growth of data amount. Approximate query processing (AQP) is a way to meet the requirement of fast response. In this paper, we propose a learning-based AQP method called the LAQP. The LAQP builds an error model learned from the historical queries to predict the sampling-based estimation error of each new query. It makes a combination of the sampling-based AQP, the pre-computed aggregations and the learned error model to provide high-accurate query estimations with a small off-line sample. The experimental results indicate that our LAQP outperforms the sampling-based AQP, the pre-aggregation-based AQP and the most recent learning-based AQP method.
\end{abstract}


\end{frontmatter}

\section{Introduction}\label{sec:introduction}

It is a difficult task to obtain the exact query answers on big data. Even though sufficient hardware is available to conduct queries on big data, hours of response time is unacceptable to make real-time decisions~\cite{DBLP:journals/debu/MozafariN15, DBLP:conf/sigmod/PengZWP18}. Approximate query processing (AQP)~\cite{DBLP:journals/tods/ChaudhuriDN07, DBLP:conf/sigmod/ChaudhuriDK17} makes it possible to efficiently obtain approximate query results. The AQP has been studied for a long time. Different methods make different trade-offs among the accuracy, the response time, the space budget and the supported queries~\cite{DBLP:journals/dase/LiL18}. However, it is still challenging to achieve a satisfactory trade-off uniting all these aspects~\cite{DBLP:conf/sigmod/PengZWP18, DBLP:journals/dase/LiL18}.

Existing methods adopt two kinds of ideas.
The first one is the sampling-based AQP methods, i.e., the query results are estimated based on random samples. Online sampling methods collect samples in the process of conducting each query, which will increase the response time accordingly~\cite{DBLP:journals/dase/LiL18}. The off-line sample is collected before executing the queries, which leads to fast response. However, a small sample may not be adequately to represent the entire dataset~\cite{DBLP:journals/tods/ChaudhuriDN07, DBLP:conf/sigmod/LiHYLS08, DBLP:journals/jsw/InoueKG16}.
The second one is based on pre-computed synopses or aggregations, which are computed before query processing and used to estimate query results. However, these methods can only support some special queries. They cannot support queries as general as those supported by the sampling methods~\cite{DBLP:journals/dase/LiL18}. In addition, sufficient pre-computed aggregations will cost much space~\cite{DBLP:conf/sigmod/AgarwalMKTJMMS14}. A limited number of pre-computed aggregations or data cubes are difficult to provide sufficiently accurate estimation results, especially for the high-dimensional data.

Recently, some methods make efforts to adopt machine learning methods to solve the AQP problems~\cite{DBLP:conf/sigmod/MaT19, DBLP:conf/icde/OlmaPAA19}.
The main drawback of these methods is that the query estimation based on the learned model cannot provide a priori error guarantee as the sampling-based methods~\cite{DBLP:conf/sigmod/MaT19, DBLP:conf/icde/OlmaPAA19, DBLP:journals/corr/abs-1903-10000}. However, the error guarantee is an important metric to measure the quality of the AQP estimations~\cite{DBLP:conf/sigmod/AgarwalMKTJMMS14}.

Most of the existing methods have some shortcomings in different aspects including the low accuracy, unmeasured error, large space requirement, limited query support, and low efficiency.
Clearly, an ideal AQP approach achieves high accuracy of the estimations compared to the sampling-based AQP methods, provides error guarantee like the sampling-based method, costs little space, supports general queries and responds fast to the queries. Thus, to support general queries and provide the error guarantee, we take full advantage of the sampling-based method. Furthermore we also adopt the pre-computed aggregations to increase the accuracy according to the previous work AQP++~\cite{DBLP:conf/sigmod/PengZWP18}. 

The AQP++ estimates the new query based on its  `range-similar' pre-computed query, whose predicate range is the most similar to the given query.
It estimates a new query $q_{new}$ as follows.

AQP++:
\begin{displaymath}
  q_{new}=q_{old}+(\hat{q}_{new}-\hat{q}_{old})
\end{displaymath}
The estimation of $q_{new}$ is the sum of two parts. The first part is the exact result of the old query $q_{old}$. The second part is the difference between the new query and the old one estimated by sampling. The $\hat{q}_{new}$ and the $\hat{q}_{old}$ are the estimations of $q_{new}$ and $q_{old}$ based on sampling. The estimation accuracy depends on the accuracy of the second part, since the first part is a constant. If the $q_{new}$ and $q_{old}$ have similar predicate range, suggesting that, they are highly correlated, the new query estimated based on the old query is possible to be more reliable than the sampling-based estimation $\hat{q}_{new}$. Therefore, it estimates a new query based on its `range-similar' pre-computed query. That is the main idea of the AQP++.

Our LAQP use a strategy different from the AQP++ to find a proper pre-computed query for each new query. In order to introduce the idea of our method, we further refine the second part of the above equation as follows.
\begin{displaymath}
  \hat{q}_{new}-\hat{q}_{old}=[q_{new}-Error(q_{new})]+[q_{old}-Error(q_{old})]
\end{displaymath}

If (1) the estimation error of these two queries $Error(q_{new})=q_{new}-\hat{q}_{new}$ and $Error(q_{old})=q_{old}-\hat{q}_{old}$ are close to each other, and (2) the  $\hat{q}_{new}$ and the $\hat{q}_{old}$ are computed based on the same sample, estimating the new query based on the old query is still reliable even though the range predicates of the two queries are irrelevant. We call thus old query the `error-similar' query of the new query.

We proved that is likely to provide a more accurate estimation for a new query based on its `error-similar' pre-computed query. Assuming that we have a query log storing the true result of each per-computed query, the precondition for finding the `error-similar' aggregation is to predict the sampling-based error for each new query. We involve machine learning models in error prediction, since they are useful to predict the unknowns with the past observations. We calculate the estimations of the pre-computed queries based on a fixed sample. It is also easy to compute the estimation errors, since we know the true results of the pre-computed queries. Thus, a regression model can be learned by mapping each pre-computed query's predicate to its estimation error which is computed based on the fixed sample. We describe that model as \textit{mapping each query to its sampling-based estimation error} for brevity in the rest of the paper.  


In this way, we unit the advantages of the sampling-based AQP method, the pre-computed aggregations and the machine learning. At the same time, we avoid the shortcomings of them including the high space cost of a sufficient sample, and the lack of error guarantee for the predictions of a machine learning model.

For avoiding too much space cost, we use a small off-line sample to estimate all the queries. This sample is clearly insufficient to provide an accurate estimation. With the consideration that the error model learned from the query log can predict the sampling-based error of each query, we measure the quality of the small sample with such a model.  Meaning that, we can still give an accurate estimation of the difference between a given query and a pre-computed one according to their sampling-based estimations and errors. We proved that the accuracy of estimating their difference determines the accuracy of the final estimation result. In this way, we can provide a sufficiently accurate estimation with only a small sample.

For supporting an error guarantee, we use a machine learning regression model, such as the SVM, the RandomForest, and the ANN, to predict the sampling-based estimation error of a query instead of directly estimating the query result. The model can be tuned by mapping each pre-computed query to its sampling-based estimation error. In this way, the LAQP benefits from the regression model to choose an `error-similar' pre-computed query resulting in less error. Meanwhile, the estimation error can still be limited according to the statistical theorems. In addition, the fast prediction of the model will not cost much response time.



The framework of our LAQP is shown in Figure~\ref{Fig:Framework}. We randomly choose a small sample from the dataset. Note that this is the only sample we used in our LAQP. We then compute the estimation of each query in the given query log based on the small sample in order to measure the quality of the sample. We store the sampling-based estimations and the errors of estimations along with the results of queries in the log. We train a regression model mapping each query to its sampling-based estimation error. When a new query arrives, its sampling-based estimation error is predicted according to the error model. Its `error-similar' pre-computed query $Q_{opt}$ will then be used to estimate the new query. The final query result is the sum of the chosen pre-computed query $R_{opt}$ and the difference between the new query and the pre-computed one estimated by the same sample $(\hat{R}_{new}-\hat{R}_{opt})$. Since the difference is estimated based on the sample, the estimation error can be limited according to statistical theorems such as the Central Limit Theorem.

\begin{figure}
\centering
\includegraphics[scale=0.5]{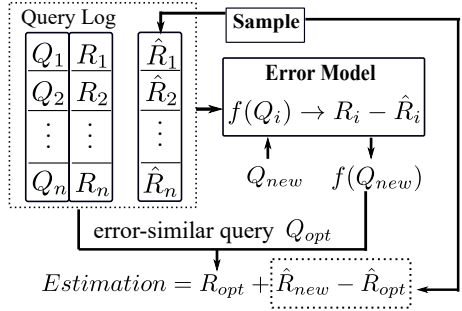}
\caption{The framework of LAQP}
\label{Fig:Framework}
\end{figure}

We make the following contributions in this paper.

(1) Our first contribution is the learning-based AQP method (LAQP). We build a regression model learned from the given query log to predict the sampling-based estimation error of each new query. To the best of our knowledge, this is the first work making a combination of the sampling-based AQP method, the pre-computed aggregations, and the regression model to increase the accuracy, limit estimation error, support general queries while occupying little space. 

(2) Our second contribution is the error analysis of the estimations and two extensions of the LAQP benefiting from the diversification and the optimization. Our LAQP can provide an error guarantee for each estimation according to the statistical theorems. The diversification and optimization methods can be applied to our LAQP to improve its performance.

(3) Our third contribution is the extensive experiments. We compare the performance of our method with some existing representative approaches including a sampling-based AQP method, a pre-aggregation-based method and a recent learning-based AQP method called DBEst. The experimental results indicate the advantages of our method over the existing ones.

\section{Related works}
The sampling-based AQP methods are widely used, due to its efficiency and universality. 
Usually, there is a tradeoff between the efficiency and the accuracy~\cite{DBLP:conf/cidr/SidirourgosKB11, DBLP:conf/sigmod/PengZWP18, DBLP:journals/dase/LiL18}. Samples in small size respond fast at the cost of reducing the accuracy. Many works make efforts to find a representative sample in a small size to improve accuracy without reducing the efficiency~\cite{DBLP:conf/cidr/SidirourgosKB11, DBLP:conf/eurosys/AgarwalMPMMS13, DBLP:journals/tods/ChaudhuriDN07}. The main drawback of random sampling is that the accuracy decreases with the variance of the aggregated attribute values. That is, random sampling cannot provide sufficiently accurate estimations for the attributes with high-skew distribution. Stratified sampling~\cite{DBLP:conf/sigmod/AcharyaGP00, DBLP:conf/vldb/GantiLR00, DBLP:journals/tods/ChaudhuriDN07, DBLP:conf/cidr/SidirourgosKB11, DBLP:conf/eurosys/AgarwalMPMMS13} is a way to solve this problem. The probability of stratified sampling is related to the importance or contribution to the aggregation result. However, the stratified sampling usually relies on a prior knowledge of the distribution. 

Making use of some pre-computed synopses and aggregations is another way to solve this problem. Synopses such as histograms~\cite{DBLP:journals/tods/GibbonsMP02}, sketches~\cite{DBLP:conf/sigmod/RoyKA16}, and wavelets \cite{DBLP:conf/sigmod/GarofalakisG02} can be computed according to the query workload before executing queries. Pre-computed aggregations are also based on the given query workload. Approximate Pre-Aggregation~\cite{DBLP:conf/vldb/Jermaine03} (APA) uses a random sample combined with a small set of statistics about the data to increase the accuracy. There are some methods storing some data cubes as pre-aggregations to improve the accuracy~\cite{DBLP:conf/vldb/Dyreson96, DBLP:conf/sigmod/MumickQM97}. These method cannot support queries as general as those supported by the sampling-based AQP methods. There are some other methods combining the pre-aggregations with sampling~\cite{DBLP:conf/dmkd/JermaineM00, DBLP:conf/sigmod/LiHYLS08, DBLP:conf/icde/KamatJTN14, DBLP:journals/tkdd/KamatN18, DBLP:conf/sigmod/PengZWP18}. The authors of reference ~\cite{DBLP:conf/sigmod/PengZWP18} proposed the AQP++ estimating the query result based on the pre-computed data cubes and the sampling-based AQP.

Machine learning methods have been used in the field of data processing and data analysis in recent years~\cite{DBLP:conf/sigmod/KraskaBCDP18, DBLP:journals/corr/abs-1903-10000}. There are some new AQP methods adopt machine learning methods~\cite{DBLP:conf/sigmod/MaT19, DBLP:conf/icde/OlmaPAA19}. In reference~\cite{DBLP:conf/sigmod/MaT19}, the authors proposed an AQP method learning a density model and a regression model from a small sample. Some researches use deep generative models to learn the data distribution and generate samples for the AQP~\cite{DBLP:journals/corr/abs-1903-10000}. The main problem of using machine learning methods to approximate query result is that it currently does not provide a priori error guarantee like the sampling-based AQP methods.

Most of the existing methods have some shortcomings in different aspects including the low accuracy, limited query support, the storage of sufficient samples or pre-computed information and the low efficiency. Our LAQP supports most of the typical queries supported by the sampling-based AQP method. In the meanwhile, it aims at increasing the accuracy with a small sample. 

\section{Preliminaries}\label{section:preliminaries}
In this section, we introduce two representative AQP methods, i.e., the sampling-based AQP and the  pre-aggregation-based AQP.

\subsection{Sampling-based AQP}\label{section:SAQP}
Assuming that there is an aggregation query in the following form.

$q: $ SELECT $SUM(A)$ from $D$ WHERE $C$.

A sampling-based AQP method first randomly chooses a subset $S$ from the dataset $D$. The query approximation including the estimation result and the confidence interval are computed based on the query results on the sample $S$ and the Central Limit Theorem.
\begin{displaymath}
  EST(q)=\frac{|D|}{|S|} SUM(S_{C}(A))\pm\lambda\sqrt{\frac{var(S_{C}(A))}{|S|}}
\end{displaymath}

The $S_{C}$ means the tuples in the sample $S$ matching the predicate $C$. 

\subsection{Pre-aggregation-based AQP}
The pre-aggregation-based AQP estimates a new query based on a pre-computed query. The final estimation result $est(q)$ is the sum of the pre-computed aggregation $Pre(Q)$ of query $Q$ and the estimated difference $(EST(q)-EST(Q))$.
\begin{displaymath}
  est(q)=Pre(Q)+(EST(q)-EST(Q)),
\end{displaymath}
where $EST(q)$ and $EST(Q)$ are the estimated results of the query $q$ and $Q$ based on the method in Section~\ref{section:SAQP}.

\section{LAQP}\label{section:LAQP}
We will introduce the framework of our LAQP in Section~\ref{section:LAQPframework}. In Section~\ref{section:LAQPguarantee}, we analyze the estimation error.

\subsection{The framework}\label{section:LAQPframework}

In this section, we will introduce the framework of our LAQP.
LAQP randomly chooses a small sample from the data, and estimates the queries in the given log based on the sample. We then train a regression model from the query log to predict the sampling-based error for the given query. For each query, we search the log for its `error-similar' pre-computed query resulting in less estimation error. At last, the query result is estimated according to the sum of the chosen pre-computed query result and the difference between the new query and the pre-computed one. The estimation error can be limited according to the statistical theorems such as CLT, and Hoeffding bounds.

As introduced in Section 1, the major tasks of LAQP include providing the error guarantee by statistics, saving the space cost by maintaining a small off-line sample, and increasing the accuracy by making use of the pre-aggregations and a regression model. We will introduce how our LAQP achieves these tasks in detail as follows.

The first task is to measure the error of each estimation. Our idea of the LAQP starts from the combination of the pre-computed aggregations and the sampling-based AQP, which estimates a new query according to a pre-computed query result and the difference between the new query and the old one. The accuracy of the difference estimated by sampling determines the accuracy of the query estimation. The estimation error is limited according to the statistical theorems.

 Our second task is to narrow down the space cost. We store a small off-line sample in the memory, since it costs little and responds fast. However, choosing a good sample in a limited size is a difficult task. If the distribution of the data is skew or the number of dimensions increases to a large number, the distribution of a small sample has little chance to be the same with that of the entire data. Instead of finding a perfect sample, we would like to keep a small static sample and measure the quality of that sample according to the pre-computed results in the query logs.

The third task is to improve the accuracy. The accuracy of the pre-aggregation-based AQP method depends on the accuracy of estimating the difference between the new query and a proper pre-computed one as we mentioned before. We come up with a new idea of choosing the `error-similar' pre-computed query instead of the `range-similar' pre-computed one. We will prove that this idea leads to a more accurate estimation. The motivation of this idea is composed of two parts. First, the assumption that the estimation errors of two `range-similar' queries are similar is not quite reliable for the high-skew data. Even though we can find a `range-similar' pre-computed one, is that the optimal choice? There may still exist another `error-similar' pre-computed one resulting in the higher accuracy. Second, how to choose a proper one from a small number of pre-computed queries whose predicate ranges are all not sufficiently similar to the new query? The probability of choosing a sufficiently similar predicate is not high, because of the increasing dimensions and the limited number of the pre-computed queries. In that situation, the `error-similar' query is possible to be a good choice.


The remaining problem is how to compute the sampling-based estimation error of a new query before execution. The previous method estimates a new query based on the `range-similar' pre-computed query, since the predicate range is the only information of the new query while the sampling-based error of a new query is unknown. We make use of a machine learning approach to solve this problem. The assumption of our method is that the pre-computed queries and their true results are available in the query log. We compute the estimations of the pre-computed queries in the given query log based on the same sample, and learn a regression model mapping each query to the difference between its true result and its estimation. We call the difference the sampling-based estimation error. The predictive ability of the regression model makes it possible to predict the estimation error of a large number of new queries while occupying a little space. Furthermore, its ability of handling multi-dimensional data benefits the multi-dimensional AQP. In addition, predicting an error with a simple model is efficient, since the machine learning models are widely used in real-time decision making.

\begin{algorithm}
\caption{LAQP-ModelConstruction}\label{algorithm:LAQP-ModelConstruction}
\fontsize{8pt}{0}
\textbf{Input: QueryLog $QL=\{[Q_1,R_1], [Q_2,R_2], ..., [Q_n,R_n]\}$, DATA $D$ }

\textbf{Output: Sample $S$, Error model $f$}
\begin{algorithmic}[1]
\State $S\leftarrow$ Uniform Random Sample from $D$
\For {$Q_i$ in $QL$}
    \State $\hat{R}_i = SAQP(Q_i,S)$
\EndFor
\State Training Model $f: Q_i\rightarrow R_i-\hat{R}_i$
\end{algorithmic}
\end{algorithm}

\begin{algorithm}
\caption{LAQP-Estimation}\label{algorithm:LAQP-Estimation}
\fontsize{8pt}{0}
\textbf{Input: QueryLog $QL$, Error model $f$, Sample $S$, New query $q$ }

\textbf{Output: Estimation $est$}
\begin{algorithmic}[1]
\State $PredictedError\leftarrow f(q)$
\State $opt=\arg\min_i |(R_i-\hat{R}_i)-PredictedError|$
\State $est=R_{opt}+SAQP(q,S)-SAQP(Q_{opt},S)$
\end{algorithmic}
\end{algorithm}

Our LAQP is composed of two parts, i.e., error model learning, and query estimation. The pseudo-code of them are shown in Algorithm~\ref{algorithm:LAQP-ModelConstruction} and Algorithm~\ref{algorithm:LAQP-Estimation}, respectively.

The first part is to construct the error model. It first randomly choose a small sample $S$ from the dataset (Line 1). The sample is then used to compute the estimation of each query $Q$ in the QueryLog according to the simplest SAQP (Line 2-4). $SAQP(Q_i, S)$ denotes the approximate query result of $Q_i$ estimated with the sample $S$ based on the method introduced in Section~\ref{section:SAQP} . The algorithm then trains an error model $f$ mapping each query in the query log to its estimation error (Line 5). In the implementation, we train one model for one kind of aggregation query. For example, the model of the sum queries like \textit{Q: select sum(X) from D, where $l_A\le A \le r_A$ and $l_B\le B \le r_B$} can be formed as $f_{sum}: (l_A,r_A,l_B,r_B)\rightarrow sum(X)-\hat{sum}(X)$. For brevity, we uniformly represent the model as the form in the pseudo-code.

The second part is to estimate the result of a new query $q$. First, the error model predicts the error $f(q)$ of the new query (Line 1). We regard a query $Q_i$, whose estimation error $(R_i-\hat{R}_i)$ is the closest to the predicted error $f(q)$ of the new query, as the baseline pre-computed aggregation (Line 2). The final estimation result is calculated by summing the pre-computed result $R_{opt}$ and the estimated difference $SAQP(q, S) - SAQP(Q_{opt}, S)$ between the new query $q$ and the pre-computed $Q_j$ (Line 3).

We show the process of our algorithm with the following example.

\begin{table}
\centering
\begin{tabular}{|c|c|c|c|c|c|c|}
\hline
Price&120&195&200&210&250&280\\
\hline
Count&75&10&15&20&15&65\\
\hline
Sample&4&1&1&1&1&2\\
\hline
\end{tabular}
\caption{Information in an order list.}\label{Table:Example1a}
\end{table}

\begin{table}
\centering
\begin{tabular}{|c|c|c|c|c|}
\hline
&Range&True&Est&Error\\
\hline
$Q_1$&$[100,200]$&100&120&$+20$\\
\hline
$Q_2$&$[201,300]$&100&80&$-20$\\
\hline
$Q_{new}$&$[190,270]$&60&80&+16(predicted)\\
\hline
\end{tabular}
\caption{Estimate new query based on the query log.}\label{Table:Example1b}
\end{table}

\begin{example}
Table~\ref{Table:Example1a} shows the $Price$ and $Count$ of 200 items, and we choose 10 items from the entire list to estimate the queries.
Suppose the following two queries are in the query log:

$Q_1$:
\begin{minipage}[t]{0.9\linewidth}
   select count(item) from order,\\ where price between 100 and 200;\\
\end{minipage}

$Q_2$:
\begin{minipage}[t]{0.9\linewidth}
   select count(item) from order,\\ where price between 201 and 300;\\
\end{minipage}

The true query results, the sampling-based estimated results, and the estimation errors are shown in Table~\ref{Table:Example1b}. Assuming that an error model has already been well tuned based on the information in the query log. Currently, the task is to estimate the following new query:

$Q_{new}$:
\begin{minipage}[t]{0.9\linewidth}
 select count(item) from order,\\ where price between 190 and 270;\\
\end{minipage}

We can learn from the figure that the predicted error of the new query is more similar to the error of $Q_1$, suggesting that, estimating the result based on $Q_1$ is better than $Q_2$. We compute the estimations based on $Q_1$ and $Q_2$, and compare them with the true result to verify the idea of our algorithm. \\
Estimation result based on $Q_1$: $R_1=100+80-120=60$.\\
Estimation result based on $Q_2$: $R_2=100+80-80=100$.

The estimation based on $Q_1$ is more accurate, even though the predicate range of the new query is more similar to $Q_2$.
\end{example}

As shown in the example, estimating a query based on an error-similar pre-computed query is possible to provide a more accurate estimation. In the next section, we will prove that the estimation error of our LAQP method is limited.

\subsection{Error Analysis}\label{section:LAQPguarantee}

In this section, we describe the query estimation in LAQP and prove that the estimation error is limited. At last, we discuss the impacts of the estimation accuracy.

We first define the estimation of a query $q$.

\begin{definition}\label{def:estimation}
  The estimation of a query $q$ based on a pre-computed $q_i$ is defined as $est(q)$, and calculated with the following equation.
  \begin{align}
    est(q)=R(q_i)+EST(q)-EST(q_i)
  \end{align}
\end{definition}

Then, we prove that the estimation $est(q)$ in Definition~\ref{def:estimation} is the unbiased estimation of query $q$.

\begin{theorem}\label{theorem:UnbiansedEstimation}
    $est(q)$ is the unbiased estimation of $R(q)$, i.e., $\mathbb{E}[est(q)]=R(q)$.
\end{theorem}

\begin{proof}
\begin{align}
\mathbb{E}[est(q)]&=\mathbb{E}[R(q_i)+EST(q)-EST(q_i)]\\
                  &=\mathbb{E}[EST(q)]+\mathbb{E}[R(q_i)]-\mathbb{E}[EST(q_i)]
\end{align}
Since $EST(q_i)$ is the estimated by sampling,
\begin{align}
    \mathbb{E}[R(q_i)]=\mathbb{E}[EST(q_i)]
\end{align}
Consequently,
\begin{align}
    \mathbb{E}[est(q)]=\mathbb{E}[EST(q)]=\mathbb{E}[R(q)]=R(q)
\end{align}

 That is, $est(q)$ is the unbiased estimation of $R(q)$.
\end{proof}

The estimation error $R(q)-est(q)$ can be bounded according to the following theorem.
\begin{theorem}\label{theorem:error}
 $ \Pr[R(q)-est(q)>\delta\cdot R(q)]\le e^{-\delta^2\cdot R(q)/2}$, $\delta\in(0,1)$.
\end{theorem}

\begin{proof}
  By the Chernoff bound, for any $\delta\in(0,1)$,
  \begin{align}
    \Pr[est(q)<(1-\delta)\mathbb{E}[est(q)]]\le e^{-\delta^2\cdot \mathbb{E}[est(q)]]/2}.
  \end{align}

  Since $\mathbb{E}[est(q)]=R(q)$ according to Theorem~\ref{theorem:UnbiansedEstimation},

  \begin{displaymath}
    \Pr[R(q)-est(q)>\delta\cdot R(q)]\le e^{-\delta^2\cdot R(q)/2}.
  \end{displaymath}
  That is, the estimation error $R(q)-est(q)$ is limited.
\end{proof}

In this way, the estimation error of each query based on our LAQP method is limited.

In the above discussions, the $q_i$ could be any query in the query log $Q$. However, the selection of $q_i$ also influences the accuracy of the estimation. We will show the influences of the accuracy in the following Theorem~\ref{theorem:accuracy}

\begin{theorem}\label{theorem:accuracy}
 $ \min |R(q)-est(q)|=\min |PredictionError(q)+(f(q)-Error(q_i))| $, where the model $f$ maps each query $q_i$ in the query log $Q$ to its sampling-based estimation error.  $f: q_i \rightarrow Error(q_i)=R(q_i)-EST(q_i)$.
\end{theorem}

\begin{proof}
Considering the prediction error of the model,
\begin{align}
f(q)+PredictionError(q)=R(q)-Est(q).
\end{align}

The estimation error can be calculated as,
\begin{align*}
|R(q)-est(q)|&=|R(q)-[R(q_i)+EST(q)-EST(q_i)]|\\
            &=|[R(q)-EST(q)]-[R(q_i)-EST(q_i)]|\\
            &=|[f(q)+PredictionError(q)]-Error(q_i)|\\
            &=|PredictionError(q)+(f(q)-Error(q_i))|               
\end{align*}
That is, the estimation accuracy depends on two parts. The first one is the model accuracy, the second is the similarity between the sampling-based estimation error of the chosen pre-computed query and the predicted error of the new query.
\end{proof}

The prediction error mostly depends on the reliability of the training data. Once the model is trained, the only thing to do is to find a pre-computed query in the log whose sampling-based estimation error is the most similar to the predicted error. Therefore, we finally define the estimator of our LAQP as follows.

\begin{definition}\label{def:LAQP}
  The estimation of a query $q$ based on the LAQP is defined as $est(q)$, and computed according to the following equation.
  \begin{align}
    est(q)=R(q_{opt})+EST(q)-EST(q_{opt}),
  \end{align}
  where
    $q_{opt}=\arg\min_{q_i} |f(q)-[R(q_i)-EST(q_i)]|$.
\end{definition}

As we discussed above, the estimation error of LAQP is limited, and the estimation accuracy depends on both the model accuracy and the error-similarity between the chosen pre-computed query and the new query. 

\subsection{Aggregation functions}
LAQP can support any typical aggregation functions supported by the sampling-based AQP method, such as the COUNT, SUM, AVG, STD and VAR. This point can be proved easily similar to the Lemma 1 in the reference~\cite{DBLP:conf/sigmod/PengZWP18}. LAQP is able to provide an error guarantee for these aggregation functions according to the Theorem~\ref{theorem:error}.

Since the queries involving MAX and MIN are very sensitive to rare large or small values~\cite{DBLP:conf/sigmod/AgarwalMKTJMMS14}, the sampling-based estimations are not reliable for them. Similarly, LAQP cannot provide the error guarantee for the queries involving the MAX/MIN aggregations without the explicit distribution of the data. These queries depend on the rank order of the tuples rather than their actual values~\cite{DBLP:conf/icde/ChaudhuriDMN01}. However, LAQP has more information besides the sample to rely on, since it learns the sampling-based error from the pre-computed queries. Therefore, it is possible to give better estimations of the MAX/MIN queries compared with the sampling-based method. The following theorems demonstrate that the estimations of the MAX/MIN queries based on the LAQP is possible to be more accurate than the sampling-based estimation. 

\begin{theorem}\label{theorem:MAX}
  The LAQP estimation $est(q)$ of a $MAX$ query $q$ is more accurate than the sampling-based estimation $EST(q)$ as long as there exist a pre-computed query $Q$ that $Error(Q)\le 2\cdot Error(q)$, where $Error(Q)$ and $Error(q)$ are the sampling-based errors of $Q$ and $q$, respectively.
\end{theorem}

\begin{proof}
 The estimation of $q$ based on LAQP is $R(Q)+EST(q)-EST(Q)$ according to the Theorem~\ref{theorem:UnbiansedEstimation}.
  Thus, the estimation error is:
  \begin{align*}
|R(q)-est(q)|&=|R(q)-[R(Q)+EST(q)-EST(Q)]|\\
            &=|[R(q)-EST(q)]-[R(Q)-EST(Q)]|\\
            &=|Error(q)- Error(Q)|
\end{align*}
The $Error(q)$ is the difference between the true result $R(q)$ and its estimation $EST(q)$ calculated based on the sample, i.e., $Error(q)=R(q)-EST(q)$.
It is clearly that the sampling-based estimation $EST(q)$ of a $MAX$ query is the under-estimation. Therefore, $|Error(q)|=Error(q)$. Consequently, $|Error(q)|\ge|R(q)-est(q)|$ can be true on the assumption that $Error(q)\ge|Error(q)- Error(Q)|$.

In the next step, we prove that $Error(q)\ge|Error(q)- Error(Q)|$ is true as long as $Error(Q)\le 2\cdot Error(q)$. (1) If $Error(q)\ge Error(Q)$, the assumption is always true since $Error(Q)\ge0$ for any $MAX$ query. (2) If $Error(q)< Error(Q)$, the assumption is true as long as $\frac{1}{2}\cdot Error(Q) \le Error(q)< Error(Q)$. Thus, the assumption $Error(q)\ge|Error(q)- Error(Q)|$ is true as long as $\frac{1}{2}\cdot Error(Q) \le Error(q)$.
\end{proof}

\begin{theorem}\label{theorem:MIN}
  The LAQP estimation $est(q)$ of a $MIN$ query $q$ is more accurate than the sampling-based estimation $EST(q)$ as long as there exist a pre-computed query $Q$ that $Error(Q)\ge 2\cdot Error(q)$, where $Error(Q)$ and $Error(q)$ are the sampling-based errors of $Q$ and $q$, respectively.
\end{theorem}

We omit the proof of Theorem~\ref{theorem:MIN} since it can be proved in the same way as proving the Theorem~\ref{theorem:MAX}.
These theorems demonstrate that, as long as we can find a pre-computed query satisfying the assumption, the LAQP is more accurate than the sampling-based AQP.

\section{Extensions}
In this section, we discuss the diversification and optimization of LAQP to improve the performance furthermore.


\subsection{Diversification}\label{section:diversification}
In the previous section, we have discussed training the error model based on the given query logs. For such approach, we have two questions. Are all the queries in the log suitable to be used to train the error model? Do we need to store all the queries that keep coming up?   

Clearly, it is not a good idea to store all the processed queries and train the error model with them, due to its huge storage requirement.
The simplest way to reduce the storage requirement is to limit the number of queries in the log. However, reducing the number of pre-computed queries will affect the accuracy of the LAQP. We need to find a proper subset of the log that performs well with a limited number of queries. As discussed in Section~\ref{section:LAQPguarantee}, the estimation accuracy of LAQP depends on (1) the accuracy of the error model and (2) the similarity between the error of a pre-computed query and the predicted error of the given query. We consider to improve these two points by diversifying the training data and the errors of a limited number of pre-computed queries.

On the one hand, the error model benefits from the diversification of the training data.
The diversity of training data ensures that it can provide more discriminative information to the model~\cite{DBLP:journals/access/GongZH19}. Therefore, we consider increasing the diversify of the training data in LAQP, i.e., the pre-computed queries in the log.
On the other hand, the diversification of the sampling-based estimation errors is possible to improve the estimation accuracy. We show that diversifying the sampling-based errors of the pre-computed queries is possible to reduce the maximum error-difference between a query and its `error-similar' pre-computed query. That difference determines the estimation accuracy as we discussed before. Thus, the sampling-based estimation error should also be diversified.

We take the Max-Min~\cite{DBLP:conf/www/GollapudiS09} diversification method as an example to show that it is possible to reduce the estimation error. The Max-Min method aims at maximizing the minimum of the distance between any two elements in the diversified result.
We use the following example to show that the maximum of the distance between the given query and its `error-similar' pre-computed query can be reduced by the Max-Min diversification method. 

\begin{figure}
\centering
\includegraphics[scale=0.3]{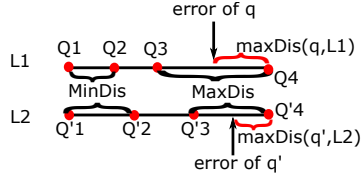}
\caption{Diversification}
\label{Fig:ExampleDiversification}
\end{figure}

\begin{example}
  Suppose that the $L_1$ and $L_2$ in Figure~\ref{Fig:ExampleDiversification} are two subsets of the query log. The red points on the first line are the sampling-based errors of the four queries $Q_1$, $Q_2$, $Q_3$, and $Q_4$ in $L_1$. The red points on the second line are those errors in $L_2$. Assuming that the ranges of the errors in $L_1$ and $L_2$ are the same, and $L_2$ is the diversified result chosen from the query log, the minimum distance (MinDis) of two errors in $L_2$ is larger than that of the distance in $L_1$. The max distance $maxDis(q,L_1)$ between the error of a query $q$ and its nearest error in $L_1$ is the half of the max distance (MaxDis) in $L_1$. Since the MinDis in $L_2$ is larger than that of $L_1$, the maximum of its MaxDis is smaller than that of $L_1$, i.e., $maxDis(q,L_1)>maxDis(q,L_2)$.
\end{example}

As shown in the above example, the diversification method is possible to reduce the maximum of the difference between a query and its `error-similar' pre-computed one in the log. The difference determines the accuracy of the query estimation as we discussed before. Thus, diversifying the sampling-based errors of the pre-computed queries is possible to reduce the estimation error.



To discuss the diversification strategy, we define the distance of two queries as follows.

\begin{definition}
 The distance of query $Q_i$ and $Q_j$:
 \begin{displaymath}
   Dis(Q_i, Q_j)=\frac{\sum_{x \in d}^{(l_{i,x}-l_{j,x})^2+(r_{i,x}-r_{j,x})^2}}{2d}+(Error_i-Error_j)^2,
 \end{displaymath}
 where $d$ is the number of dimensions, the $l_{i,x}$ and $r_{i,x}$ mean the left and right range boundaries of the $Q_i$ in the $x$th dimension, and the $Error_i$ means the sampling-based estimation error of $Q_i$.
\end{definition}

According to the definition of the distance between two queries, a variety of diversification methods~\cite{DBLP:conf/www/GollapudiS09, DBLP:conf/icde/VieiraRBHSTT11} can be adopted to find a set of queries with high diversity. The usage of the diversification makes it possible to tune a good model with limited pre-computed aggregations and increases the opportunity of finding a sufficiently similar pre-computed query for each new query. 
Note that the normalization is required before the computation of the distance between two multi-dimensional queries in the implement.
We make experiments to test the improvement of the LAQP benefiting from the diversification in Section~\ref{exp:diversification}.

\subsection{Optimization}\label{section:optimization}
The final estimation provided by our LAQP is based on an existing pre-computed query as discussed before. However, the pre-computed query chosen according to the `error-similar' strategy is not always quite reliable. What if the situation that the model is not well tuned? What if there exists a pre-computed query whose predicate range is sufficiently similar to the query range of the new one? That is, the LAQP may fail to find the true `error-similar' query based on the wrong prediction of the error model. We still have a chance to give an accurate estimation based on a `range-similar' pre-computed query whose predicate range is close to that of the new query. Therefore, we consider to make a proper combination of our `error-similar' strategy with the previous `range-similar' strategy to find the optimal pre-computed query when the model accuracy is not satisfactory.

\begin{algorithm}
\caption{Optimized-LAQP}\label{algorithm:LAQP-Opt}
\textbf{Input: QueryLog $QL$, DATA $D$, New Query $q$ }\\
\textbf{Output: Sample $S$, Error model $f$, Estimation $est$}
\begin{algorithmic}[1]
\State $f: Q_i\rightarrow R_i-\hat{R}_i$, $R_i\in QL$
\State $Error_q\leftarrow f(q)$
\State $Dis(q,Q_i)=\alpha\cdot EDis(q,Q_i)+\beta\cdot RDis(q,Q_i)$
\State $opt=\arg_i\min Dis(q,Q_i) $
\State $est=R_{opt}+SAQP(q,S)-SAQP(Q_{opt},S)$
\end{algorithmic}
\end{algorithm}

We involve the range-distance to the query-distance to find the similar pre-computed query. We put different weights $\alpha$ and $\beta$ on the error-distance and the range-distance, and redefine the distance of two queries as follows.
\begin{equation}\label{equation:OptDis}
  Dis(Q_i, Q_j)=\alpha\cdot EDis(Q_i, Q_j)+\beta\cdot RDis(Q_i, Q_j)
\end{equation}
 We propose a new algorithm adopting the new distance in Equation~\ref{equation:OptDis} to find the nearest pre-computed query.
 The weights in Equation~\ref{equation:OptDis} show the reliability of the error model.
They can be computed by solving an optimization problem whose object function is to minimize the error on the testing set.

The pseudo-code of our optimized-LAQP is shown in Algorithm~\ref{algorithm:LAQP-Opt}. Learning the error model is the first step (Line 1). The way of choosing the optimal pre-computed query is based on the redesigned distance of two queries involving both the error-distance($EDis$) and the range-distance($RDis$) (Line 3-4). 
The parameters $\alpha$ and $\beta$ can be either determined by the users or tuned by optimizing the accuracy.

As discussed in Section~\ref{section:LAQPguarantee}, the accuracy of an estimation is determined by the difference between the sampling-based estimation error of the new query and the chosen pre-computed one, i.e., $|error_q-error_{opt}|$. Therefore we regard the sum of that difference of each query in the $Test$ set as the object function $z$. The problem of choosing the optimal weights $\alpha$ and $\beta$ can be solved by the optimization problem as follows.

\begin{align}
\label{equation:opt}&\min\,\,z=\sum_{q\in Test} |error_q-error_{opt(q,\alpha,\beta)}|^2\\
\label{equation:Pre-Aggr}&opt(q,\alpha,\beta)=\mathop{\arg\min}_{Q_i\in Train} \alpha\cdot EDis(q,Q_i)+\beta\cdot RDis(q,Q_i)\\
&EDis(Q_i, Q_j)=|error_{Q_i}-error_{Q_j}|^2\\
&RDis(Q_i, Q_j)=\frac{\sum_{x \in d}^{(l_{i,x}-l_{j,x})^2+(r_{i,x}-r_{j,x})^2}}{2d}\\
&s.t.\,\,\alpha+\beta=1, 0\le\alpha\le 1, 0\le\beta\le 1.
\end{align}

Such an optimization problem could be solved by the approaches such as the Brent's method~\cite{brent2013algorithms}, and the weights $\alpha$ and $\beta$ are obtained in the solution.
We proved that the optimization improves the performance of the original LAQP.

\begin{theorem}
  If $\alpha=1$, the accuracy of the optimized LAQP is the same with the original LAQP. If $\alpha<1$, the accuracy of the optimized LAQP is better than the original LAQP on the testing set.
\end{theorem}

\begin{proof}
(1) If $\alpha=1$, Equation~\ref{equation:Pre-Aggr} is simplified as $opt(q,1,0)=\mathop{\arg\min}_{Q_i\in Train} |error_{q}-error_{Q_i}|$, which is the same with the equation in Algorithm 2 (Line 2). Thus, the optimized LAQP is the same with the original LAQP.

(2) If $\alpha<1$, and assuming that the optimized LAQP is worse than the original LAQP, the sum of the squared error on the testing set of the optimized LAQP is higher than that of the original LAQP, i.e.,  $\sum_{q\in Test} |error_q-error_{opt(q,\alpha,\beta)}|^2 > \sum_{q\in Test} |error_q-error_{opt(q,1,0)}|^2$. It contradicts the object function of the optimization in Equation~\ref{equation:opt}. Therefore, the assumption is false. Thus, the optimized LAQP is better than the original LAQP in this situation.
\end{proof}

Therefore, the optimization improves the accuracy on the testing set. Involving the optimization only modifies the way of choosing the pre-computed query for each new query, the estimation error is still limited according to the Theorem~\ref{theorem:error} as we discussed in Section~\ref{section:LAQPguarantee}.


\section{Experiments}
In this section, we show the experimental results of the LAQP. We compare the LAQP with the existing methods including the most typical and widely used sampling-based AQP method, the most recent method combining the pre-computed aggregations with sampling called AQP++, and the state-of-the-art learning-based AQP method called DBEst.

\subsection{Experimental Setup}


\noindent \underline{Hardware and Library}
All the experiments were conducted on a laptop with an Intel Core i5 CPU with 2.60GHz clock frequency and 8GB of RAM.
The error model implementation is based on the RandomForestRegressor in the scikit-learn library \footnote{\url{https://scikit-learn.org/stable/}}. 

\noindent \underline{Datesets}
We use three real-life datasets for experiments.The POWER and WESAD are two large multi-dimensional datasets. The distribution of the attribute `global\_active\_power' in POWER is a long-tailed distribution, while the distribution of each attribute in WESAD approximates a normal distribution. We also use the PM2.5 dataset adopted in the experiments of the DBEst for fair comparisons.

(1) The POWER dataset is the ``Individual household electric power consumption Data Set''\footnote{\url{http://archive.ics.uci.edu/ml/datasets/Individual+household+electric+power+consumption}}. This dataset contains 2,075,259 tuples and 9 attributes. We use a subset of this dataset including seven numerical attributes and 2,000,000 tuples to conduct the experiments.

(2) The WESAD (Wearable Stress and Affect Detection) dataset~\cite{DBLP:conf/icmi/SchmidtRDML18} is a real-life dataset. This dataset is a 16GB dataset containing 63 million records. We use eight attributes (CH1, CH2, CH3, CH4, CH5, CH6, CH7, CH8) from it to conduct our experiments.

(3) The PM2.5 dataset is a real-life hourly dataset containing the information of the PM2.5 data of the US Embassy in Beijing \footnote{\url{http://archive.ics.uci.edu/ml/datasets/Beijing+PM2.5+Data}}. There are 43,824 instances in this dataset.   

\noindent \underline{Queries}
The queries in our experiments include multi-dimensional predicates for testing the influence of dimension on the performance. We make examples about the queries adopted in the experiments.

\noindent An example of a one-dimensional query:

$Q_{1D}:$
\begin{minipage}[t]{0.9\linewidth}
  SELECT COUNT(pm2.5) from PM2.5,\\ where $l\le$PREC$\le r$.
\end{minipage}

\noindent An example of a three-dimensional query:

$Q_{3D}:$
\begin{minipage}[t]{0.9\linewidth}
SELECT COUNT(CH1) from WESAD,\\ where $l_1\le$CH1$\le r_1$, $l_2\le$CH2$\le r_2$, $l_3\le$CH3$\le r_3$.
\end{minipage}

We introduce the generation of the queries in the experiments as follows.

(1) The aggregated attribute in the queries for the PM2.5 dataset is `pm2.5', and each predicate only involves one attribute `PREC'.
The one-dimensional queries are generated by randomly choosing the range boundaries from the domain of the attribute `PREC'.

(2) The aggregated attribute for the POWER dataset is `global\_active\_power', and each predicate involves seven attributes.
In order to avoid most of the multi-dimensional query results to be zero, we limit the range to generate the boundaries in each dimension. The left boundary of the range in each dimension is randomly chosen from the first quarter of the entire value range of the corresponding attribute. Similarly, each right boundary is randomly chosen from the last quarter of the corresponding attribute range.

(3) The aggregated attribute for the WESAD dataset is `CH1', and each predicate involves eight attributes.
The predicates are generated in the same way as those generated for the POWER dataset.

For the following experiments, we generate the queries for each dataset. We take a subset of the queries to form the query log, and estimate the remaining queries based on the LAQP method.


\noindent \underline{Error Metrics}
We use two error metrics to measure the accuracy of our LAQP method. In the following definitions, $Q$, $\hat{r_i}$ and $r_i$ denote the queries, the estimated query results and the true results, respectively.

Average Relative Error: 
$ARE = \frac{1}{|Q|}\sum_{q_i\in Q}\frac{|\hat{r_i}-r_i|}{r_i}$

Mean Squared Error: 
$MSE = \frac{1}{|Q|}\sum_{q_i\in Q}|\hat{r_i}-r_i|^2$


\noindent \underline{Implementation Details}
All the experiments are conducted in Python 3.5. We use the RandomForestRegressor in the scikit-learn package to model the relation between each pre-computed query and its sampling-based estimation error. A random forest is a meta estimator that fits a number of decision tree classifiers on various sub-samples of the dataset and uses averaging to improve the predictive accuracy and control over-fitting~\cite{DBLP:journals/ml/Breiman01}. We simply adopted the RandomForestRegressor since it is widely used and its parameters are easy to tune. 
We set the parameter of the RandomForestRegressor $max\_depth=3$. This parameter is determined by tuning. We use the first dataset to test the performance of the model with different $max\_depth$ and show the result in Figure~\ref{fig:RF}.

\begin{figure}
	\centering
	\subfigure[MSE]{
        \includegraphics[width=0.4\textwidth,height=3cm]{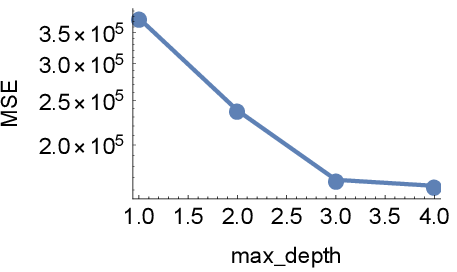}
        \label{fig:RF-MSE}
	}
    \subfigure[SpaceCost]{
		\includegraphics[width=0.4\textwidth,height=3cm]{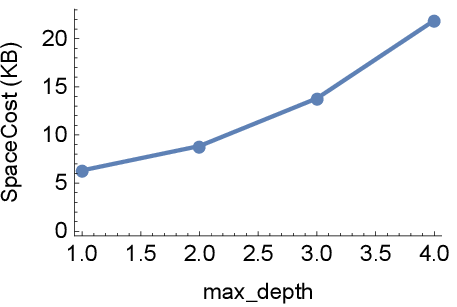}
        \label{fig:RF-SpaceCost}
	}
	\caption{The impact of the max\_depth of the RandomForest}
	\label{fig:RF}
\end{figure}




\noindent \underline{Competitors}
In the following experiments, we compare the performance of LAQP with the most representative or the state-of-the-art AQP methods including the simplest SAQP (the sampling-based AQP), a modified version of the AQP++~\cite{DBLP:journals/tkdd/KamatN18} and the DBEst~\cite{DBLP:conf/sigmod/MaT19}.

SAQP: the simplest sampling-based AQP method as introduced in Section~\ref{section:SAQP}. It estimates the query result according to a sample randomly chosen from the entire dataset.

AQP++: an AQP method combining the sampling-based AQP method with the pre-computed data cubes. In this experiment, we did not generate data cubes, but modify the AQP++ method to regard the pre-computed aggregations as the pre-computed cubes. We still follow the main idea of the original AQP++ estimating a new query based on a pre-computed query and computing the difference between the new query and the old one based on sampling. The reason for the modification is that we want to test the performance the AQP methods with a small number of pre-aggregations. However, if the BP-Cube in the original AQP++ is adopted, each dimension can be only partitioned into very few parts. In addition, the queries are not generated uniformly in the data range to avoid the query result to be zero. Thus, the queries may be far from the cells in the BP-Cube. However, since the queries and the pre-aggregations are generated in the same way as we introduced before, the pre-aggregations are more reliable than the BP-Cube. Thus, the modified version performs better than the original AQP++ in the following experiments.
In the following figures and experimental descriptions, the modified version is still called AQP++ for the sake of brevity.

DBEst: a learning-based AQP method which learns a density model and a regression model from a small sample of the data.
We adopted the implementation of the DBEst provided by the authors in github\footnote{\url{https://github.com/qingzma/DBEstClient}}.

\subsection{Accuracy}

In this section, we compare the accuracy of our LAQP method with the other methods. The accuracy of each method is measured by both the MSE and ARE. We also test the influence of the query selectivity on the accuracy. We compare our LAQP with the SAQP and AQP++ for the multi-dimensional queries. The comparisons involving the DBEst are only conducted on the one-dimensional queries due to the limitation of its implementation. We test their performance for three aggregation functions $Count$, $Sum$ and $Avg$. The queries are generated as introduced before. To be fair, all the methods use the same off-line sample.

\begin{figure}
	\centering
	\subfigure[MSE]{
        \includegraphics[width=0.4\textwidth,height=3cm]{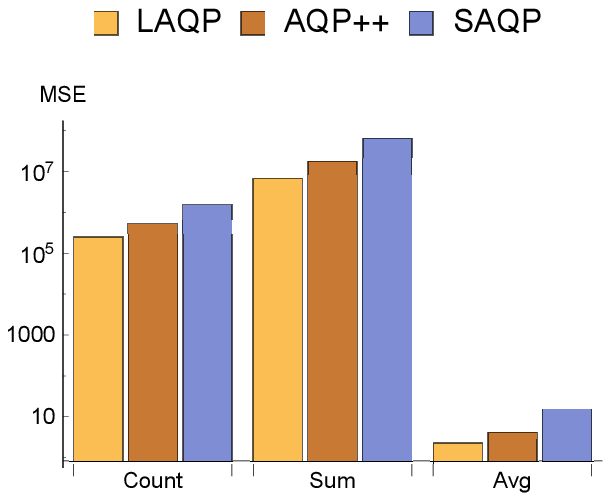}
        \label{fig:AccuracyMSE}
	}
    \subfigure[ARE]{
		\includegraphics[width=0.4\textwidth,height=3cm]{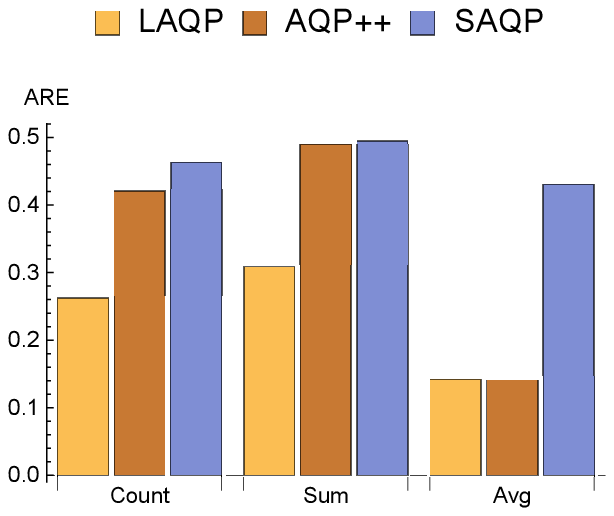}
        \label{fig:AccuracyARE}
	}
	\caption{Accuracy Comparison on Dataset POWER}
	\label{fig:Accuracy}
\end{figure}

\begin{figure}
	\centering
	\subfigure[MSE]{
        \includegraphics[width=0.4\textwidth,height=3cm]{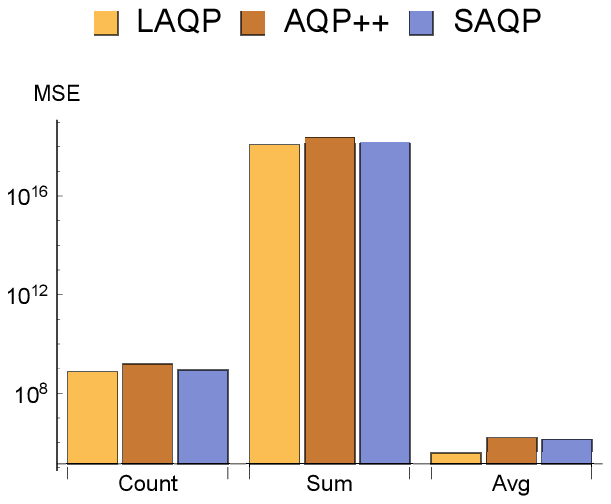}
        \label{fig:AccuracyWESADMSE}
	}
    \subfigure[ARE]{
		\includegraphics[width=0.4\textwidth,height=3cm]{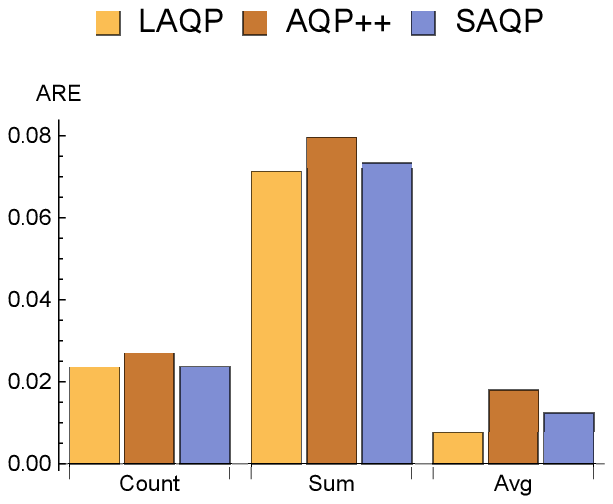}
        \label{fig:AccuracyWESADARE}
	}
	\caption{Accuracy Comparison on Dataset WESAD}
	\label{fig:AccuracyWESAD}
\vspace{-0.3cm}
\end{figure}

\begin{figure}
	\centering
	\subfigure[MSE]{
        \includegraphics[width=0.4\textwidth,height=3cm]{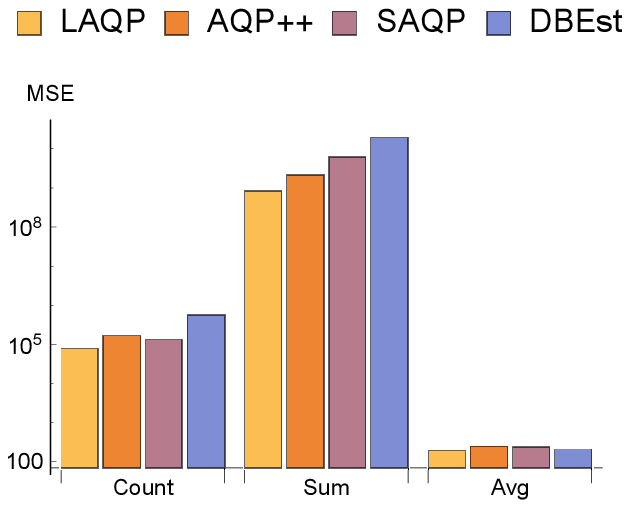}
        \label{fig:AccuracyPM25MSE}
	}
    \subfigure[ARE]{
		\includegraphics[width=0.4\textwidth,height=3cm]{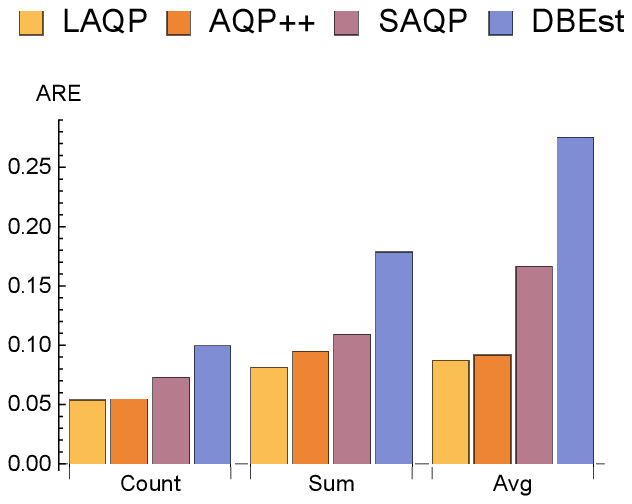}
        \label{fig:AccuracyPM25ARE}
	}
	\caption{Accuracy Comparison on Dataset PM2.5}
	\label{fig:AccuracyPM25}
\end{figure}

\noindent \underline{EXP1}:  We compare the accuracy of our LAQP, SAQP and the AQP++ on the dataset POWER. The results are shown in Figure~\ref{fig:Accuracy}. The number of samples is 2,000. We adopted 800 and 100 queries in the log to train and test the error model of the LAQP, respectively. We compute the estimation error of 100 new queries to evaluate the accuracy.  Each query includes a seven-dimensional predicate. The average selectivity of the queries is nearly 0.2\%. The results are shown in Figure~\ref{fig:Accuracy}. We can learn from the figure that the our LAQP is more accurate than the other two methods. Since the sample size is much smaller than the data size, the sample distribution is not sufficiently similar to that of the entire data. Therefore, the SAQP does not perform well in this experiment. Even though a large number of the pre-computed queries are adopted, it is still difficult for the AQP++ to find a similar pre-computed query for each new due to the high dimensions. As a comparison, LAQP does not suffer from the small sample size and the multi-dimensional data. LAQP benefits from the ability of the error model to handle the multi-dimensional data. In addition, it is easier to find a similar one-dimensional estimation error than finding a similar multi-dimensional predicate range.

\noindent \underline{EXP2}:  We compare the accuracy of our LAQP, SAQP and the AQP++ on the dataset WESAD. The results are shown in Figure~\ref{fig:AccuracyWESAD}.  The number of samples is 20K. We adopted 130 and 40 queries in the query log to train and test the error model of the LAQP, and we computed the estimation error of 30 new queries to evaluate the accuracy.  Each query includes an eight-dimensional predicate. The average selectivity of the queries is nearly 2\%.  In this experiment, the SAQP performs better than the AQP++, since the sample is sufficient while the number of the pre-computed queries is limited. Our LAQP is still more accurate than the other two methods.

\noindent \underline{EXP3}: We compare the accuracy of our LAQP, SAQP, AQP++ and DBEst on the dataset PM2.5. The experimental results are shown in Figure~\ref{fig:AccuracyPM25}. We randomly choose a sample with the sampling rate at 1\% from the entire dataset. We only use a query log including 200 one-dimensional queries in this experiment, and compute the estimation error of 100 new queries to evaluate the accuracy. The queries in this experiment only include the one-dimensional predicates due to the limitation of the implementation of the DBEst.  We can learn from the figure that the accuracy of our LAQP outperforms the other methods. The reason why our method outperforms the SAQP and AQP++ is similar to that in EXP1. The reason why LAQP outperforms the DBEst is that we use pre-computed queries to improve the accuracy, while the accuracy of the DBEst method largely depends on the sample. When the sample size is much smaller than the data size, the regression model of DBEst based on the sample is unreliable.

\noindent \underline{EXP4}: We test the influence of the query selectivity on the accuracy. The queries in the experiments only contain one-dimensional predicates and two-dimensional predicates, since it is difficult to generate high-dimensional queries with a high selectivity. The sample size of these methods are 2,000, and the number of the pre-computed queries is 200. The relative errors of the estimation results on the one-dimensional queries and the two-dimensional queries are shown in Figure~\ref{fig:Selectivity} and Figure~\ref{fig:Selectivity2D}, respectively. We can learn from the figures that the relative errors of the estimations decrease with the selectivity.
In most cases, LAQP is more accurate than the other methods.
The performance of these three methods are similar on AVG queries. The reason is that the variance of the aggregated attribute values is not high, meaning that the results of the AVG queries are similar. As a comparison, LAQP still has superiority to handle the two-dimensional AVG queries as shown in Figure~\ref{fig:SelectivityAvg2D}.
In this experiment, the sampling rate is only 0.1\% and only a small number of pre-computed queries are adopted in AQP++ and LAQP. Therefore, the distribution of the sample is not sufficiently similar to that of the entire dataset, and it is difficult to find a pre-computed query similar to each new query. Our LAQP is better than the SAQP since the sampling-based estimation errors are learned from the pre-computed queries to improve the accuracy. Our LAQP provides an opportunity to find an `error-similar' pre-computed query, other than finding a `range-similar' pre-computed one. As the dimension increases, it will be more difficult to find a  pre-computed query with sufficiently similar range. This point is also verified by the fact that the advantage of LAQP on the two-dimensional queries is more significant than that on the one-dimensional queries.

\begin{figure*}
	\centering
	\subfigure[COUNT]{
        \includegraphics[width=0.3\textwidth,height=3cm]{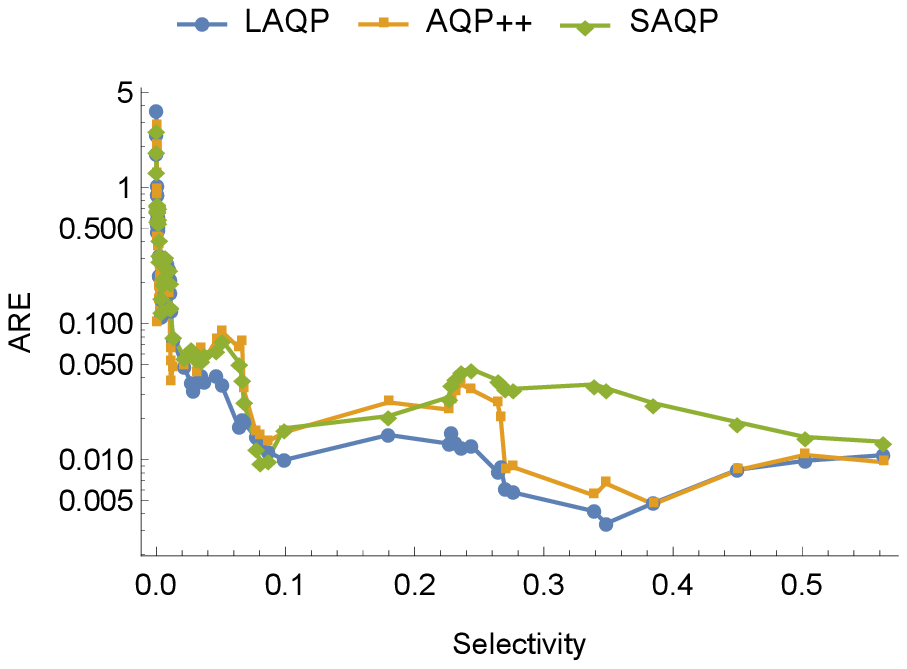}
        \label{fig:SelectivityCount}
	}
    \subfigure[SUM]{
		\includegraphics[width=0.3\textwidth,height=3cm]{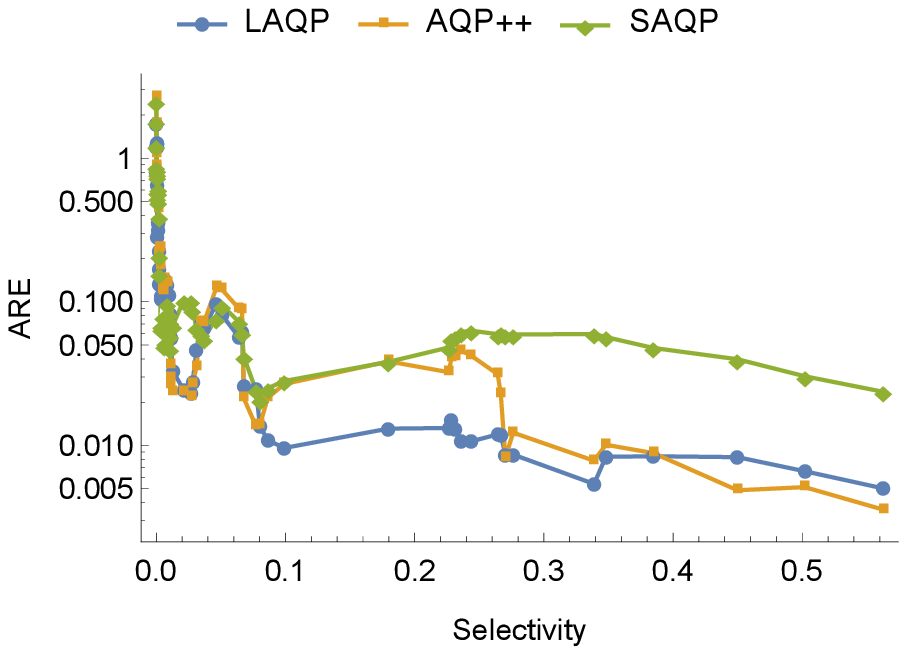}
        \label{fig:SelectivitySum}
	}
    \subfigure[AVG]{
		\includegraphics[width=0.3\textwidth,height=3cm]{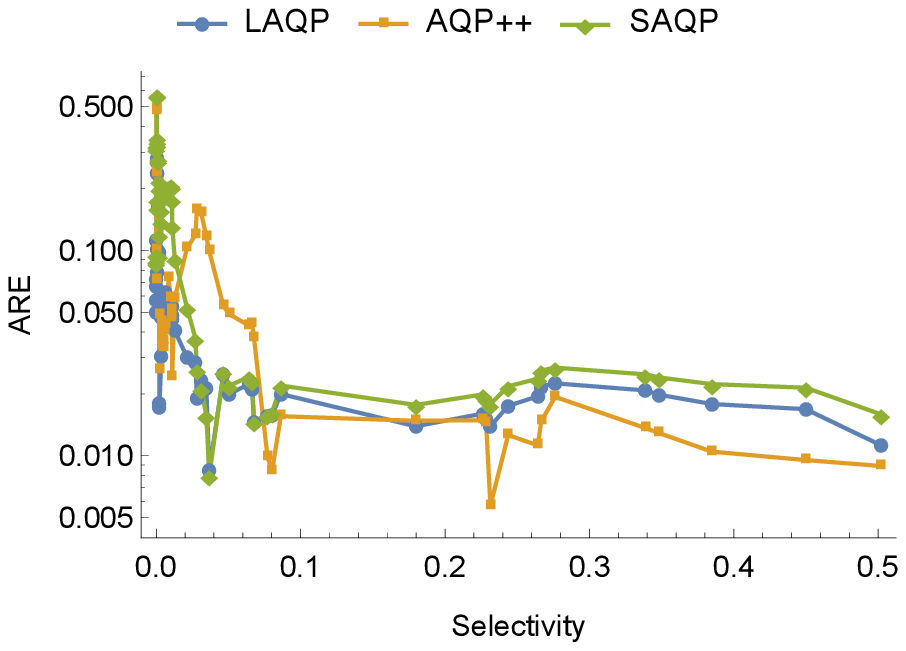}
        \label{fig:SelectivityAvg}
	}
	\caption{The influence of selectivity on the accuracy (1D).}
	\label{fig:Selectivity}
\end{figure*}

\begin{figure*}
	\centering
	\subfigure[COUNT]{
        \includegraphics[width=0.3\textwidth,height=3cm]{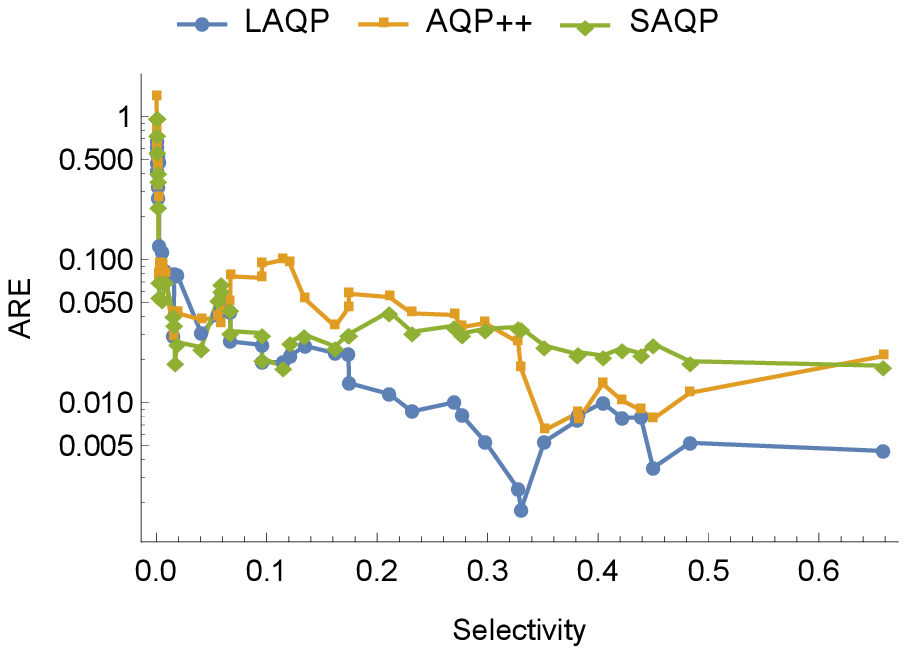}
        \label{fig:SelectivityCount2D}
	}
    \subfigure[SUM]{
		\includegraphics[width=0.3\textwidth,height=3cm]{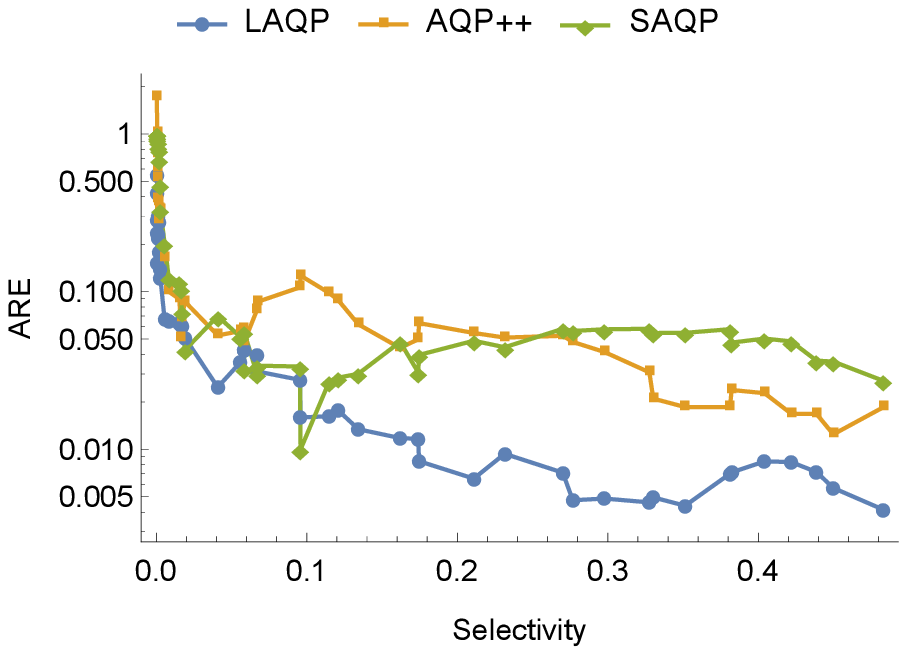}
        \label{fig:SelectivitySum2D}
	}
    \subfigure[AVG]{
		\includegraphics[width=0.3\textwidth,height=3cm]{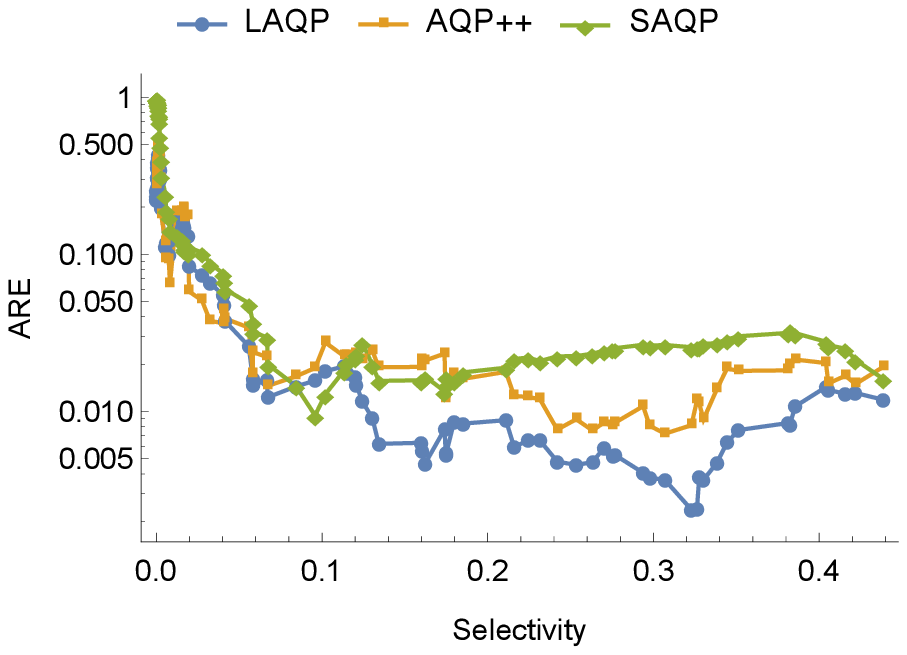}
        \label{fig:SelectivityAvg2D}
	}
	\caption{The influence of selectivity on the accuracy (2D).}
	\label{fig:Selectivity2D}
\end{figure*}

\subsection{Space Cost}
We compare the influence of the space cost on the accuracy. The space cost of the SAQP is the sample size. The space cost of the AQP++ includes both the sample size and the cost of the pre-computed aggregations. LAQP has an additional cost of the error model compared with the AQP++. Since the space cost of the sample, the pre-computed aggregations, and the error model are independent of the data size, we did not evaluate the relation of the data size and the space cost. Instead, we focus on the influence of the space cost on the accuracy.

This experiment was conducted on the POWER dataset. Each query involves a seven-dimensional predicate.
The number of the samples in the SAQP varies from 1k to 5k. Since the space cost of LAQP and AQP++ include three and two parts, respectively, it is difficult to make their space cost absolutely the same with that of the SAQP. We adopted different settings to find the space cost of LAQP and AQP++ in the range similar to that of the SAQP.
The settings (the number of samples, the number of pre-computed queries, the space cost of the error model, and the total space cost) of different methods are shown in Table~\ref{Table:SpaceCost}.  The experimental results are shown in Figure~\ref{fig:SpaceCost}.
\begin{table}
\centering
\begin{tabular}{|c|c|c|c|c|}
\hline
Method&Sample&Pre-Queries&Model&SpaceCost(KB)\\
\hline
LAQP&1000&250&12KB&184\\
\hline
LAQP&2000&250&12KB&296\\
\hline
LAQP&2000&500&13KB&358\\
\hline
LAQP&2000&800&14KB&430\\
\hline
LAQP&2000&1000&14KB&534\\
\hline
AQP++&1000&250& &172\\
\hline
AQP++&1000&500& &232\\
\hline
AQP++&2000&500& &344\\
\hline
AQP++&2000&800& &416\\
\hline
AQP++&2000&1000& &464\\
\hline
SAQP&1000& & &112\\
\hline
SAQP&2000& & &224\\
\hline
SAQP&3000& & &336\\
\hline
SAQP&4000& & &448\\
\hline
SAQP&5000& & &560\\
\hline
\end{tabular}
\caption{Settings for the space cost experiment.}\label{Table:SpaceCost}
\end{table}

We can learn from Figure~\ref{fig:SpaceCost} that both the MSE and ARE decrease with the space cost. This phenomenon is absolutely reasonable since more samples, more pre-computed queries and more complex model lead to higher accuracy. The estimation error of our method is lower than that of the SAQP and AQP++ for the most of the time when only a little space are provided. The distribution of a small sample has little chance to be similar to the entire data. Insufficient pre-computed queries also make it difficult to find a similar pre-computed query. However, the error model can learn the estimation error of a small sample while occupying little space. That is the reason why LAQP outperforms the other methods when the provided space is limited. However, when enough samples and pre-computed queries are provided, LAQP will not be better than the other two methods. 

\begin{figure}
	\centering
	\subfigure[MSE]{
        \includegraphics[width=0.4\textwidth,height=3cm]{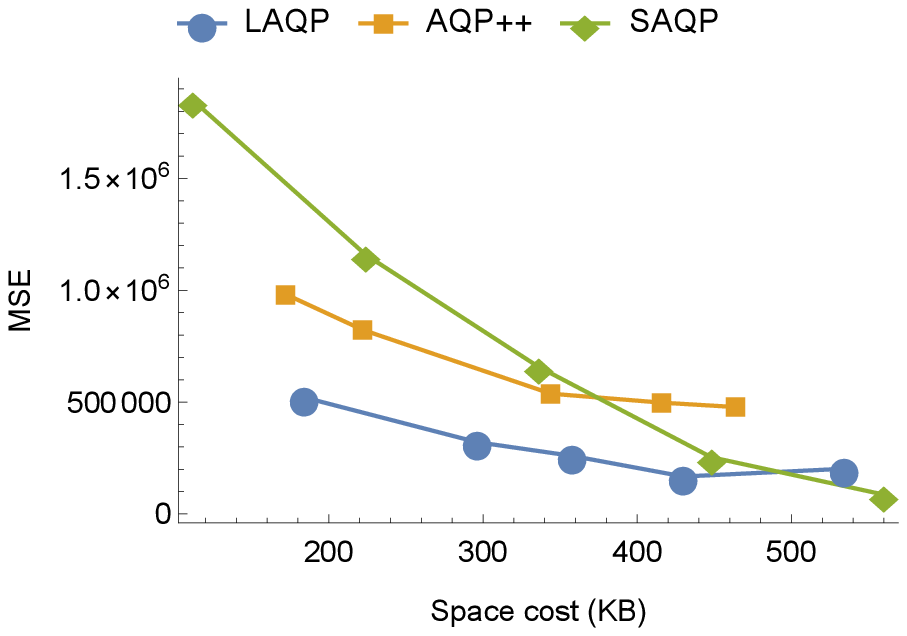}
        \label{fig:SpaceCost-MSE}
	}
    \subfigure[ARE]{
		\includegraphics[width=0.4\textwidth,height=3cm]{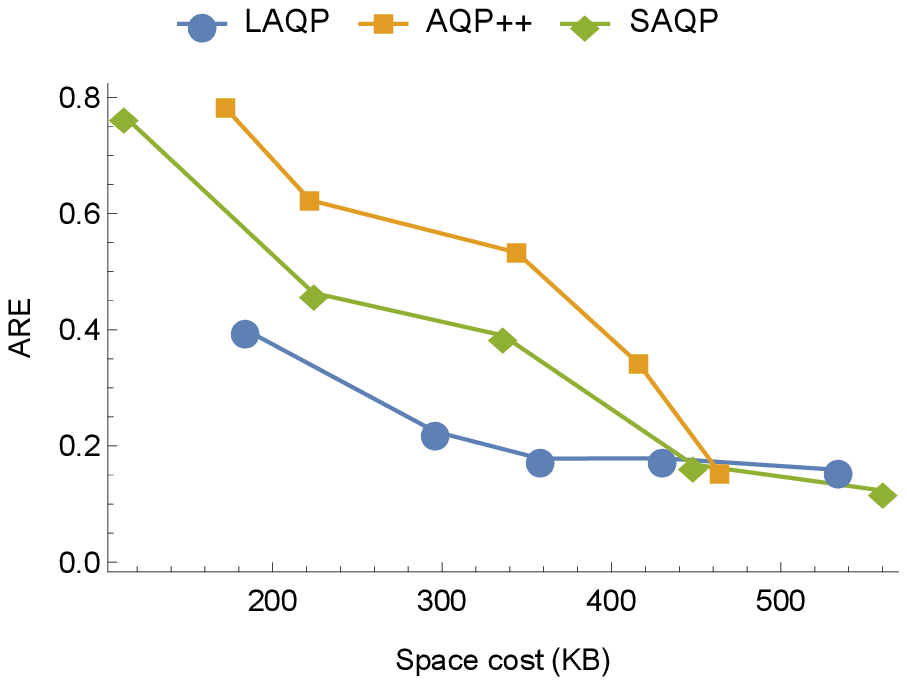}
        \label{fig:SpaceCost-ARE}
	}
	\caption{The impact of the space cost}
	\label{fig:SpaceCost}
\end{figure}

\subsection{Efficiency}
We test the efficiency of LAQP, and compare our method with the existing methods. We compare the average processing time of estimating 100 queries. We also test the influence of the dimensions on the processing time. The query time of the LAQP includes the prediction, finding the nearest query, and the sampling-based estimation. The aggregation function in these experiments is COUNT. The pre-computation and the training time are off-line process, and they are not included in the following experiments.

\noindent \underline{EXP1}: We compare the query processing time of the LAQP, AQP++, SAQP and DBEst on the PM2.5 dataset. The sample includes 4K tuples. The number of the pre-computed queries is 100. The average processing time of 100 one-dimensional queries are shown in Figure~\ref{Fig:Efficiency_PM25}. We can learn that our LAQP is comparable to the other methods, even though the processing time of LAQP includes three parts.

\noindent \underline{EXP2}: We compare the query processing time of the LAQP, AQP++, and SAQP on the POWER dataset. The query predicate varies from one dimension to seven dimensions. The sample in this experiment includes 20K tuples. We compute the processing time of 100 queries for each kind of aggregation function. The experimental results are shown in Figure~\ref{Fig:Efficiency_Dimension}. The processing time of these three methods all increases with the number of dimensions. The processing time of LAQP and AQP++ is nearly two times of the SAQP. The reason is that both the processing time of LAQP and AQP++ includes estimating the given query and a pre-computed query based on a sample. If the sampling-based estimation of each pre-computed query is computed off-line and stored in memory, the query time of them will be close to that of the SAQP. 

\begin{figure}
\begin{minipage}[h]{0.5\linewidth}
\center
\includegraphics[width=1\linewidth,height=4cm]{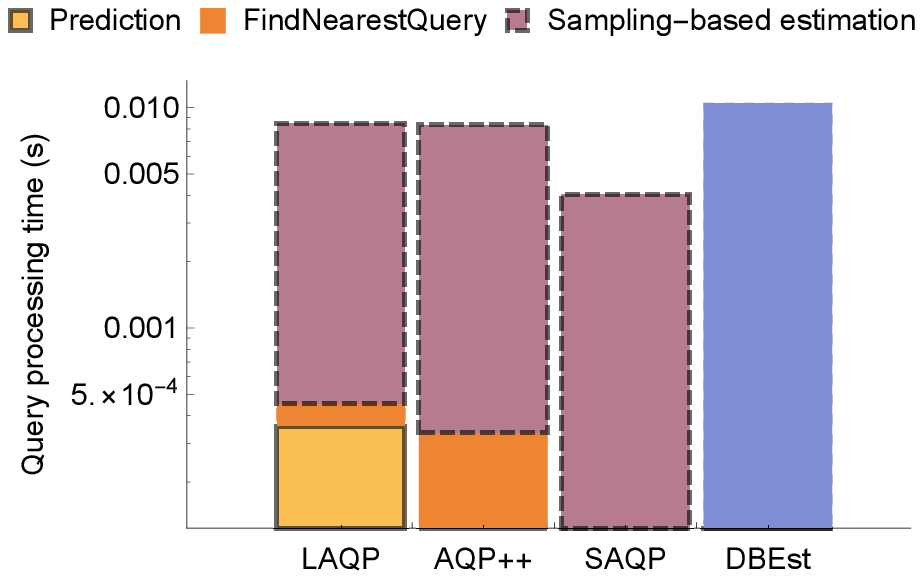}
\caption{Efficiency}
\label{Fig:Efficiency_PM25}
\end{minipage}
\begin{minipage}[h]{0.5\linewidth}
\center
\includegraphics[width=1\linewidth,height=4cm]{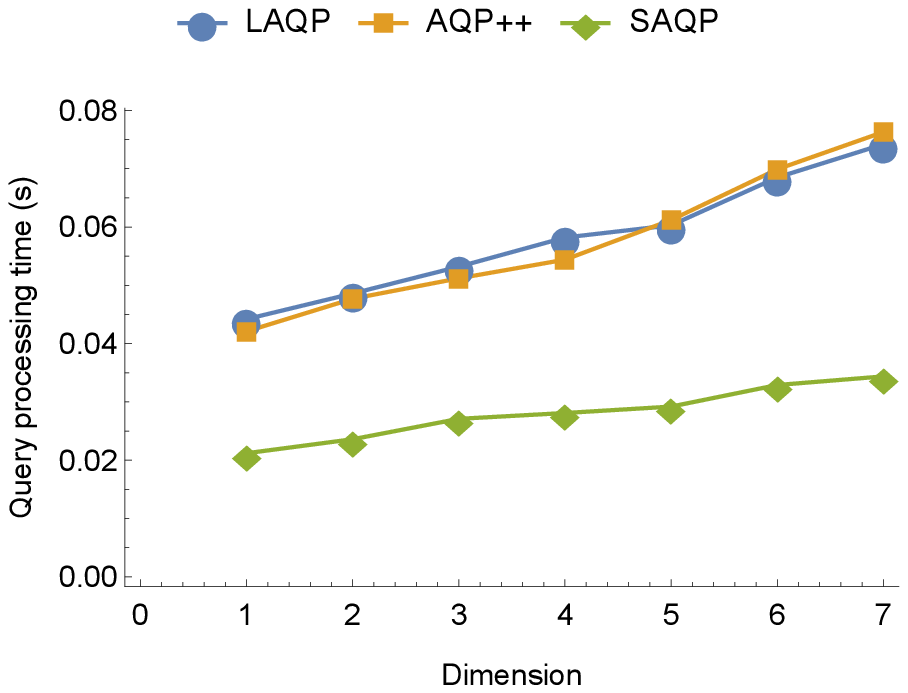}
\caption{The impact of dimension}
\label{Fig:Efficiency_Dimension}
\end{minipage}
\end{figure}

\subsection{Other aggregation functions}\label{exp:aggregations}
In this section, we test the performance of LAQP on the other aggregation functions including the VAR, STD, MIN and MAX. We compare our LAQP with the SAQP and AQP++ on the PM2.5 dataset. One hundred one-dimensional pre-computed queries for each kind of aggregation function are used to train the error model. The result of this experiment is shown in Figure~\ref{fig:Aggregation}. The result of the queries involving the MIN aggregation function is nearly 0 for most of the time, suggesting that, even a little absolute error becomes a high relative error. We can learn from the figure that our LAQP has better performance for most of the time.

\begin{figure}
	\centering
	\subfigure[MSE]{
        \includegraphics[width=0.4\textwidth,height=4cm]{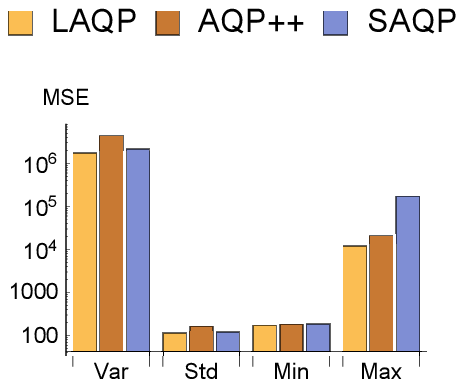}
        \centering
        \label{fig:Aggregation-MSE}
	}
    \subfigure[ARE]{
		\includegraphics[width=0.4\textwidth,height=4cm]{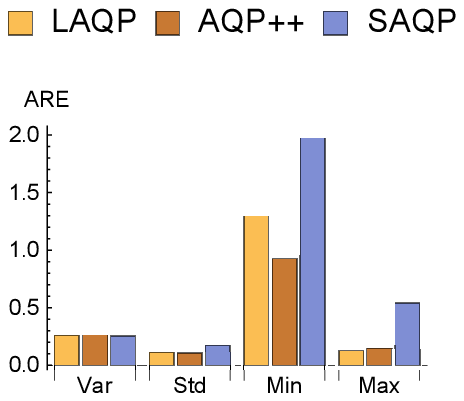}
        \centering
        \label{fig:Aggregation-ARE}
	}
	\caption{Different aggregation functions.}
	\label{fig:Aggregation}
\end{figure}

\subsection{The Benefits from the Diversification}\label{exp:diversification}
We test the improvement of LAQP by diversification. 
We use a randomly chosen query log including 200 one-dimensional queries and a diversified query log including 200 queries to train the error model, respectively. The diversification method in this experiment is the MaxMin method, which greedily inserts the query maximizing the minimum distance to the existing queries into the diversified log. This experiment was conducted on the PM2.5 dataset. The performance of these two situations are shown in Figure~\ref{fig:DiversifiedLAQP}. The LAQP with the diversified query log is called the `DiversifiedLAQP' for short in the figure. This figure indicates that the diversified query log improves the accuracy of the LAQP, which is coincident with the analysis in Section~\ref{section:diversification} .

\begin{figure}
	\centering
	\subfigure[MSE]{
        \includegraphics[width=0.4\textwidth,height=4cm]{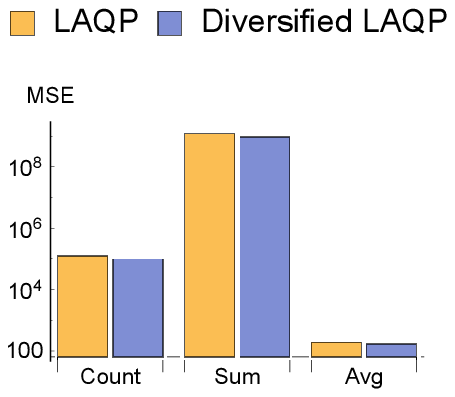}
        \label{fig:DiversifiedLAQP-MSE}
	}
    \subfigure[ARE]{
		\includegraphics[width=0.4\textwidth,height=4cm]{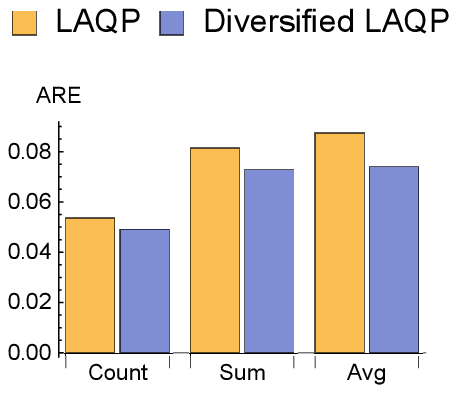}
        \label{fig:DiversifiedLAQP-ARE}
	}
	\caption{LAQP vs. DiversifiedLAQP}
	\label{fig:DiversifiedLAQP}
\end{figure}

\subsection{The Benefits from the Optimization}\label{exp:optimization}
We test the improvement of LAQP by optimization. This experiment was conducted on the PM2.5 dataset. We modify the $max\_depth$ of the RandomForest to form the error models with different reliability. We randomly choose 100 one-dimensional queries as the training set, and another 100 queries as the testing set. We can learn the influence of the weight $\alpha$ on the Object Function in Equation~\ref{equation:opt} from the Figure~\ref{fig:Optimization-alpha}. As discussed in Section~\ref{section:optimization}, the lower the result of the Object Function, the better performance on the testing set. When the error model is not well tuned ($max\_depth=1$), the result of the Object Function increases with $\alpha$. In this situation, the range-similar pre-computed query is more reliable than the error-similar one. The result of the Object Function decreases with the $\alpha$ when the error model is well tuned ($max\_depth=2$), since the error-similar pre-computed query is more reliable. We also test the performance on different kinds of aggregation functions. In this experiment, the optimization of the parameter $\alpha$ was computed by the method `bounded' in the scipy.optimize.minimize\_scalar. As shown in Figure~\ref{fig:Optimization-ARE}, optimizing the value of $\alpha$ really improves the accuracy.

\begin{figure}
	\centering
	\subfigure[Optimize $\alpha$]{
        \includegraphics[width=0.4\textwidth,height=4cm]{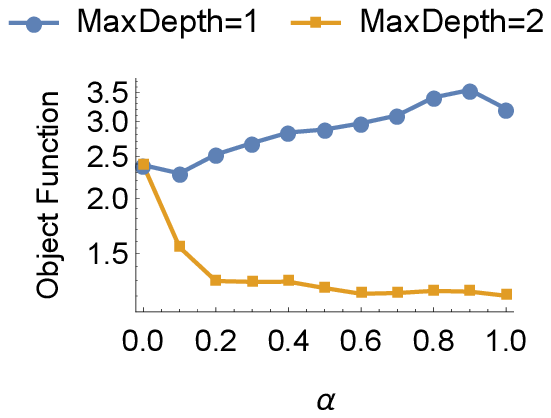}
        \label{fig:Optimization-alpha}
	}
    \subfigure[ARE]{
		\includegraphics[width=0.45\textwidth,height=4cm]{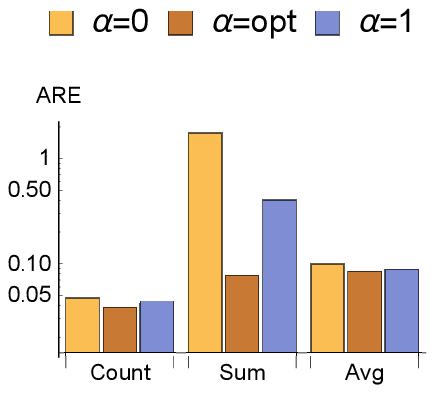}
        \label{fig:Optimization-ARE}
	}
	\caption{Optimized LAQP}
	\label{fig:Optimization}
\end{figure}

\subsection{The summary of the experiments}
From above experiments, we have following conclusions.

1. The accuracy of our LAQP outperforms the SAQP, AQP++ and DBEst when the sample are limited in a small size.

2. The superiority of LAQP to handle multi-dimensional data is evident.

3. The efficiency of LAQP is comparable to other AQP methods.

4. The performance of LAQP can be improved by diversification and optimization.

5. The LAQP performs well for most typical aggregation functions including the COUNT, SUM, AVG, VAR, STD, MIN and MAX.

\section{Conclusions}
In this work, we proposed a learning-based AQP method. We make a combination of a regression model, the sampling-based AQP and the pre-computed aggregations to provide more accurate approximate query answers. Our LAQP supports most of the typical queries supported by the sampling-based method. The performance of LAQP can be improved by involving the diversification and optimization. The experimental results indicate that our method outperforms the representative exiting methods including the sampling-based AQP method, the pre-computed aggregations based method, and the most recent learning-based method. In this work, we learn the error model based on a given query log. We will try to automatically generate the pre-computed synopses for AQP in our future study.



\bibliography{mybibfile}

\begin{thebibliography}{32}
\expandafter\ifx\csname natexlab\endcsname\relax\def\natexlab#1{#1}\fi
\providecommand{\url}[1]{\texttt{#1}}
\providecommand{\href}[2]{#2}
\providecommand{\path}[1]{#1}
\providecommand{\DOIprefix}{doi:}
\providecommand{\ArXivprefix}{arXiv:}
\providecommand{\URLprefix}{URL: }
\providecommand{\Pubmedprefix}{pmid:}
\providecommand{\doi}[1]{\href{http://dx.doi.org/#1}{\path{#1}}}
\providecommand{\Pubmed}[1]{\href{pmid:#1}{\path{#1}}}
\providecommand{\bibinfo}[2]{#2}
\ifx\xfnm\relax \def\xfnm[#1]{\unskip,\space#1}\fi
\bibitem[{Acharya et~al.(2000)Acharya, Gibbons and
  Poosala}]{DBLP:conf/sigmod/AcharyaGP00}
\bibinfo{author}{Acharya, S.}, \bibinfo{author}{Gibbons, P.B.},
  \bibinfo{author}{Poosala, V.}, \bibinfo{year}{2000}.
\newblock \bibinfo{title}{Congressional samples for approximate answering of
  group-by queries}, in: \bibinfo{booktitle}{Proceedings of the 2000 {ACM}
  {SIGMOD} International Conference on Management of Data, May 16-18, 2000,
  Dallas, Texas, {USA.}}, pp. \bibinfo{pages}{487--498}.
\newblock \URLprefix \url{https://doi.org/10.1145/342009.335450},
  \DOIprefix\doi{10.1145/342009.335450}.
\bibitem[{Agarwal et~al.(2014)Agarwal, Milner, Kleiner, Talwalkar, Jordan,
  Madden, Mozafari and Stoica}]{DBLP:conf/sigmod/AgarwalMKTJMMS14}
\bibinfo{author}{Agarwal, S.}, \bibinfo{author}{Milner, H.},
  \bibinfo{author}{Kleiner, A.}, \bibinfo{author}{Talwalkar, A.},
  \bibinfo{author}{Jordan, M.I.}, \bibinfo{author}{Madden, S.},
  \bibinfo{author}{Mozafari, B.}, \bibinfo{author}{Stoica, I.},
  \bibinfo{year}{2014}.
\newblock \bibinfo{title}{Knowing when you're wrong: building fast and reliable
  approximate query processing systems}, in: \bibinfo{booktitle}{International
  Conference on Management of Data, {SIGMOD} 2014, Snowbird, UT, USA, June
  22-27, 2014}, pp. \bibinfo{pages}{481--492}.
\newblock \URLprefix \url{https://doi.org/10.1145/2588555.2593667},
  \DOIprefix\doi{10.1145/2588555.2593667}.
\bibitem[{Agarwal et~al.(2013)Agarwal, Mozafari, Panda, Milner, Madden and
  Stoica}]{DBLP:conf/eurosys/AgarwalMPMMS13}
\bibinfo{author}{Agarwal, S.}, \bibinfo{author}{Mozafari, B.},
  \bibinfo{author}{Panda, A.}, \bibinfo{author}{Milner, H.},
  \bibinfo{author}{Madden, S.}, \bibinfo{author}{Stoica, I.},
  \bibinfo{year}{2013}.
\newblock \bibinfo{title}{Blinkdb: queries with bounded errors and bounded
  response times on very large data}, in: \bibinfo{booktitle}{Eighth Eurosys
  Conference 2013, EuroSys '13, Prague, Czech Republic, April 14-17, 2013}, pp.
  \bibinfo{pages}{29--42}.
\newblock \URLprefix \url{https://doi.org/10.1145/2465351.2465355},
  \DOIprefix\doi{10.1145/2465351.2465355}.
\bibitem[{Breiman(2001)}]{DBLP:journals/ml/Breiman01}
\bibinfo{author}{Breiman, L.}, \bibinfo{year}{2001}.
\newblock \bibinfo{title}{Random forests}.
\newblock \bibinfo{journal}{Machine Learning} \bibinfo{volume}{45},
  \bibinfo{pages}{5--32}.
\newblock \URLprefix \url{https://doi.org/10.1023/A:1010933404324},
  \DOIprefix\doi{10.1023/A:1010933404324}.
\bibitem[{Brent(2013)}]{brent2013algorithms}
\bibinfo{author}{Brent, R.P.}, \bibinfo{year}{2013}.
\newblock \bibinfo{title}{Algorithms for minimization without derivatives}.
\newblock \bibinfo{publisher}{Courier Corporation}.
\bibitem[{Chaudhuri et~al.(2001)Chaudhuri, Das, Datar, Motwani and
  Narasayya}]{DBLP:conf/icde/ChaudhuriDMN01}
\bibinfo{author}{Chaudhuri, S.}, \bibinfo{author}{Das, G.},
  \bibinfo{author}{Datar, M.}, \bibinfo{author}{Motwani, R.},
  \bibinfo{author}{Narasayya, V.R.}, \bibinfo{year}{2001}.
\newblock \bibinfo{title}{Overcoming limitations of sampling for aggregation
  queries}, in: \bibinfo{booktitle}{Proceedings of the 17th International
  Conference on Data Engineering, April 2-6, 2001, Heidelberg, Germany}, pp.
  \bibinfo{pages}{534--542}.
\newblock \URLprefix \url{https://doi.org/10.1109/ICDE.2001.914867},
  \DOIprefix\doi{10.1109/ICDE.2001.914867}.
\bibitem[{Chaudhuri et~al.(2007)Chaudhuri, Das and
  Narasayya}]{DBLP:journals/tods/ChaudhuriDN07}
\bibinfo{author}{Chaudhuri, S.}, \bibinfo{author}{Das, G.},
  \bibinfo{author}{Narasayya, V.R.}, \bibinfo{year}{2007}.
\newblock \bibinfo{title}{Optimized stratified sampling for approximate query
  processing}.
\newblock \bibinfo{journal}{{ACM} Trans. Database Syst.} \bibinfo{volume}{32},
  \bibinfo{pages}{9}.
\newblock \URLprefix \url{https://doi.org/10.1145/1242524.1242526},
  \DOIprefix\doi{10.1145/1242524.1242526}.
\bibitem[{Chaudhuri et~al.(2017)Chaudhuri, Ding and
  Kandula}]{DBLP:conf/sigmod/ChaudhuriDK17}
\bibinfo{author}{Chaudhuri, S.}, \bibinfo{author}{Ding, B.},
  \bibinfo{author}{Kandula, S.}, \bibinfo{year}{2017}.
\newblock \bibinfo{title}{Approximate query processing: No silver bullet}, in:
  \bibinfo{booktitle}{Proceedings of the 2017 {ACM} International Conference on
  Management of Data, {SIGMOD} Conference 2017, Chicago, IL, USA, May 14-19,
  2017}, pp. \bibinfo{pages}{511--519}.
\newblock \URLprefix \url{https://doi.org/10.1145/3035918.3056097},
  \DOIprefix\doi{10.1145/3035918.3056097}.
\bibitem[{Dyreson(1996)}]{DBLP:conf/vldb/Dyreson96}
\bibinfo{author}{Dyreson, C.E.}, \bibinfo{year}{1996}.
\newblock \bibinfo{title}{Information retrieval from an incomplete data cube},
  in: \bibinfo{booktitle}{VLDB'96, Proceedings of 22th International Conference
  on Very Large Data Bases, September 3-6, 1996, Mumbai (Bombay), India}, pp.
  \bibinfo{pages}{532--543}.
\newblock \URLprefix \url{http://www.vldb.org/conf/1996/P532.PDF}.
\bibitem[{Ganti et~al.(2000)Ganti, Lee and
  Ramakrishnan}]{DBLP:conf/vldb/GantiLR00}
\bibinfo{author}{Ganti, V.}, \bibinfo{author}{Lee, M.},
  \bibinfo{author}{Ramakrishnan, R.}, \bibinfo{year}{2000}.
\newblock \bibinfo{title}{{ICICLES:} self-tuning samples for approximate query
  answering}, in: \bibinfo{booktitle}{{VLDB} 2000, Proceedings of 26th
  International Conference on Very Large Data Bases, September 10-14, 2000,
  Cairo, Egypt}, pp. \bibinfo{pages}{176--187}.
\newblock \URLprefix \url{http://www.vldb.org/conf/2000/P176.pdf}.
\bibitem[{Garofalakis and Gibbons(2002)}]{DBLP:conf/sigmod/GarofalakisG02}
\bibinfo{author}{Garofalakis, M.N.}, \bibinfo{author}{Gibbons, P.B.},
  \bibinfo{year}{2002}.
\newblock \bibinfo{title}{Wavelet synopses with error guarantees}, in:
  \bibinfo{booktitle}{Proceedings of the 2002 {ACM} {SIGMOD} International
  Conference on Management of Data, Madison, Wisconsin, USA, June 3-6, 2002},
  pp. \bibinfo{pages}{476--487}.
\newblock \URLprefix \url{http://doi.acm.org/10.1145/564691.564746},
  \DOIprefix\doi{10.1145/564691.564746}.
\bibitem[{Gibbons et~al.(2002)Gibbons, Matias and
  Poosala}]{DBLP:journals/tods/GibbonsMP02}
\bibinfo{author}{Gibbons, P.B.}, \bibinfo{author}{Matias, Y.},
  \bibinfo{author}{Poosala, V.}, \bibinfo{year}{2002}.
\newblock \bibinfo{title}{Fast incremental maintenance of approximate
  histograms}.
\newblock \bibinfo{journal}{{ACM} Trans. Database Syst.} \bibinfo{volume}{27},
  \bibinfo{pages}{261--298}.
\newblock \URLprefix \url{http://doi.acm.org/10.1145/581751.581753},
  \DOIprefix\doi{10.1145/581751.581753}.
\bibitem[{Gollapudi and Sharma(2009)}]{DBLP:conf/www/GollapudiS09}
\bibinfo{author}{Gollapudi, S.}, \bibinfo{author}{Sharma, A.},
  \bibinfo{year}{2009}.
\newblock \bibinfo{title}{An axiomatic approach for result diversification},
  in: \bibinfo{booktitle}{Proceedings of the 18th International Conference on
  World Wide Web, {WWW} 2009, Madrid, Spain, April 20-24, 2009}, pp.
  \bibinfo{pages}{381--390}.
\bibitem[{Gong et~al.(2019)Gong, Zhong and Hu}]{DBLP:journals/access/GongZH19}
\bibinfo{author}{Gong, Z.}, \bibinfo{author}{Zhong, P.}, \bibinfo{author}{Hu,
  W.}, \bibinfo{year}{2019}.
\newblock \bibinfo{title}{Diversity in machine learning}.
\newblock \bibinfo{journal}{{IEEE} Access} \bibinfo{volume}{7},
  \bibinfo{pages}{64323--64350}.
\newblock \URLprefix \url{https://doi.org/10.1109/ACCESS.2019.2917620},
  \DOIprefix\doi{10.1109/ACCESS.2019.2917620}.
\bibitem[{Inoue et~al.(2016)Inoue, Krishna and
  Gopalan}]{DBLP:journals/jsw/InoueKG16}
\bibinfo{author}{Inoue, T.}, \bibinfo{author}{Krishna, A.},
  \bibinfo{author}{Gopalan, R.P.}, \bibinfo{year}{2016}.
\newblock \bibinfo{title}{Approximate query processing on high dimensionality
  database tables using multidimensional cluster sampling view}.
\newblock \bibinfo{journal}{{JSW}} \bibinfo{volume}{11},
  \bibinfo{pages}{80--93}.
\newblock \URLprefix \url{https://doi.org/10.17706/jsw.11.1.80-93},
  \DOIprefix\doi{10.17706/jsw.11.1.80-93}.
\bibitem[{Jermaine(2003)}]{DBLP:conf/vldb/Jermaine03}
\bibinfo{author}{Jermaine, C.}, \bibinfo{year}{2003}.
\newblock \bibinfo{title}{Robust estimation with sampling and approximate
  pre-aggregation}, in: \bibinfo{booktitle}{Proceedings of 29th International
  Conference on Very Large Data Bases, {VLDB} 2003, Berlin, Germany, September
  9-12, 2003}, pp. \bibinfo{pages}{886--897}.
\newblock \URLprefix \url{http://www.vldb.org/conf/2003/papers/S26P03.pdf},
  \DOIprefix\doi{10.1016/B978-012722442-8/50083-5}.
\bibitem[{Jermaine and Miller(2000)}]{DBLP:conf/dmkd/JermaineM00}
\bibinfo{author}{Jermaine, C.}, \bibinfo{author}{Miller, R.J.},
  \bibinfo{year}{2000}.
\newblock \bibinfo{title}{Approximate query answering in high-dimensional data
  cubes}, in: \bibinfo{booktitle}{2000 {ACM} {SIGMOD} Workshop on Research
  Issues in Data Mining and Knowledge Discovery, Dallas, Texas, USA, May 14,
  2000}, pp. \bibinfo{pages}{31--36}.
\bibitem[{Kamat et~al.(2014)Kamat, Jayachandran, Tunga and
  Nandi}]{DBLP:conf/icde/KamatJTN14}
\bibinfo{author}{Kamat, N.}, \bibinfo{author}{Jayachandran, P.},
  \bibinfo{author}{Tunga, K.}, \bibinfo{author}{Nandi, A.},
  \bibinfo{year}{2014}.
\newblock \bibinfo{title}{Distributed and interactive cube exploration}, in:
  \bibinfo{booktitle}{{IEEE} 30th International Conference on Data Engineering,
  Chicago, {ICDE} 2014, IL, USA, March 31 - April 4, 2014}, pp.
  \bibinfo{pages}{472--483}.
\newblock \URLprefix \url{https://doi.org/10.1109/ICDE.2014.6816674},
  \DOIprefix\doi{10.1109/ICDE.2014.6816674}.
\bibitem[{Kamat and Nandi(2018)}]{DBLP:journals/tkdd/KamatN18}
\bibinfo{author}{Kamat, N.}, \bibinfo{author}{Nandi, A.}, \bibinfo{year}{2018}.
\newblock \bibinfo{title}{A session-based approach to fast-but-approximate
  interactive data cube exploration}.
\newblock \bibinfo{journal}{{TKDD}} \bibinfo{volume}{12},
  \bibinfo{pages}{9:1--9:26}.
\newblock \URLprefix \url{https://doi.org/10.1145/3070648},
  \DOIprefix\doi{10.1145/3070648}.
\bibitem[{Kraska et~al.(2018)Kraska, Beutel, Chi, Dean and
  Polyzotis}]{DBLP:conf/sigmod/KraskaBCDP18}
\bibinfo{author}{Kraska, T.}, \bibinfo{author}{Beutel, A.},
  \bibinfo{author}{Chi, E.H.}, \bibinfo{author}{Dean, J.},
  \bibinfo{author}{Polyzotis, N.}, \bibinfo{year}{2018}.
\newblock \bibinfo{title}{The case for learned index structures}, in:
  \bibinfo{booktitle}{Proceedings of the 2018 International Conference on
  Management of Data, {SIGMOD} Conference 2018, Houston, TX, USA, June 10-15,
  2018}, pp. \bibinfo{pages}{489--504}.
\newblock \URLprefix \url{https://doi.org/10.1145/3183713.3196909},
  \DOIprefix\doi{10.1145/3183713.3196909}.
\bibitem[{Li and Li(2018)}]{DBLP:journals/dase/LiL18}
\bibinfo{author}{Li, K.}, \bibinfo{author}{Li, G.}, \bibinfo{year}{2018}.
\newblock \bibinfo{title}{Approximate query processing: What is new and where
  to go? - {A} survey on approximate query processing}.
\newblock \bibinfo{journal}{Data Science and Engineering} \bibinfo{volume}{3},
  \bibinfo{pages}{379--397}.
\newblock \URLprefix \url{https://doi.org/10.1007/s41019-018-0074-4},
  \DOIprefix\doi{10.1007/s41019-018-0074-4}.
\bibitem[{Li et~al.(2008)Li, Han, Yin, Lee and Sun}]{DBLP:conf/sigmod/LiHYLS08}
\bibinfo{author}{Li, X.}, \bibinfo{author}{Han, J.}, \bibinfo{author}{Yin, Z.},
  \bibinfo{author}{Lee, J.}, \bibinfo{author}{Sun, Y.}, \bibinfo{year}{2008}.
\newblock \bibinfo{title}{Sampling cube: a framework for statistical olap over
  sampling data}, in: \bibinfo{booktitle}{Proceedings of the {ACM} {SIGMOD}
  International Conference on Management of Data, {SIGMOD} 2008, Vancouver, BC,
  Canada, June 10-12, 2008}, pp. \bibinfo{pages}{779--790}.
\newblock \URLprefix \url{https://doi.org/10.1145/1376616.1376695},
  \DOIprefix\doi{10.1145/1376616.1376695}.
\bibitem[{Ma and Triantafillou(2019)}]{DBLP:conf/sigmod/MaT19}
\bibinfo{author}{Ma, Q.}, \bibinfo{author}{Triantafillou, P.},
  \bibinfo{year}{2019}.
\newblock \bibinfo{title}{Dbest: Revisiting approximate query processing
  engines with machine learning models}, in: \bibinfo{booktitle}{Proceedings of
  the 2019 International Conference on Management of Data, {SIGMOD} Conference
  2019, Amsterdam, The Netherlands, June 30 - July 5, 2019.}, pp.
  \bibinfo{pages}{1553--1570}.
\newblock \URLprefix \url{https://doi.org/10.1145/3299869.3324958},
  \DOIprefix\doi{10.1145/3299869.3324958}.
\bibitem[{Mozafari and Niu(2015)}]{DBLP:journals/debu/MozafariN15}
\bibinfo{author}{Mozafari, B.}, \bibinfo{author}{Niu, N.},
  \bibinfo{year}{2015}.
\newblock \bibinfo{title}{A handbook for building an approximate query engine}.
\newblock \bibinfo{journal}{{IEEE} Data Eng. Bull.} \bibinfo{volume}{38},
  \bibinfo{pages}{3--29}.
\newblock \URLprefix \url{http://sites.computer.org/debull/A15sept/p3.pdf}.
\bibitem[{Mumick et~al.(1997)Mumick, Quass and
  Mumick}]{DBLP:conf/sigmod/MumickQM97}
\bibinfo{author}{Mumick, I.S.}, \bibinfo{author}{Quass, D.},
  \bibinfo{author}{Mumick, B.S.}, \bibinfo{year}{1997}.
\newblock \bibinfo{title}{Maintenance of data cubes and summary tables in a
  warehouse}, in: \bibinfo{booktitle}{{SIGMOD} 1997, Proceedings {ACM} {SIGMOD}
  International Conference on Management of Data, May 13-15, 1997, Tucson,
  Arizona, {USA.}}, pp. \bibinfo{pages}{100--111}.
\newblock \URLprefix \url{https://doi.org/10.1145/253260.253277},
  \DOIprefix\doi{10.1145/253260.253277}.
\bibitem[{Olma et~al.(2019)Olma, Papapetrou, Appuswamy and
  Ailamaki}]{DBLP:conf/icde/OlmaPAA19}
\bibinfo{author}{Olma, M.}, \bibinfo{author}{Papapetrou, O.},
  \bibinfo{author}{Appuswamy, R.}, \bibinfo{author}{Ailamaki, A.},
  \bibinfo{year}{2019}.
\newblock \bibinfo{title}{Taster: Self-tuning, elastic and online approximate
  query processing}, in: \bibinfo{booktitle}{35th {IEEE} International
  Conference on Data Engineering, {ICDE} 2019, Macao, China, April 8-11, 2019},
  pp. \bibinfo{pages}{482--493}.
\newblock \URLprefix \url{https://doi.org/10.1109/ICDE.2019.00050},
  \DOIprefix\doi{10.1109/ICDE.2019.00050}.
\bibitem[{Peng et~al.(2018)Peng, Zhang, Wang and
  Pei}]{DBLP:conf/sigmod/PengZWP18}
\bibinfo{author}{Peng, J.}, \bibinfo{author}{Zhang, D.}, \bibinfo{author}{Wang,
  J.}, \bibinfo{author}{Pei, J.}, \bibinfo{year}{2018}.
\newblock \bibinfo{title}{{AQP++:} connecting approximate query processing with
  aggregate precomputation for interactive analytics}, in:
  \bibinfo{booktitle}{Proceedings of the 2018 International Conference on
  Management of Data, {SIGMOD} Conference 2018, Houston, TX, USA, June 10-15,
  2018}, pp. \bibinfo{pages}{1477--1492}.
\newblock \URLprefix \url{https://doi.org/10.1145/3183713.3183747},
  \DOIprefix\doi{10.1145/3183713.3183747}.
\bibitem[{Roy et~al.(2016)Roy, Khan and Alonso}]{DBLP:conf/sigmod/RoyKA16}
\bibinfo{author}{Roy, P.}, \bibinfo{author}{Khan, A.}, \bibinfo{author}{Alonso,
  G.}, \bibinfo{year}{2016}.
\newblock \bibinfo{title}{Augmented sketch: Faster and more accurate stream
  processing}, in: \bibinfo{booktitle}{Proceedings of the 2016 International
  Conference on Management of Data, {SIGMOD} Conference 2016, San Francisco,
  CA, USA, June 26 - July 01, 2016}, pp. \bibinfo{pages}{1449--1463}.
\newblock \URLprefix \url{http://doi.acm.org/10.1145/2882903.2882948},
  \DOIprefix\doi{10.1145/2882903.2882948}.
\bibitem[{Schmidt et~al.(2018)Schmidt, Reiss, D{\"{u}}richen, Marberger and
  Laerhoven}]{DBLP:conf/icmi/SchmidtRDML18}
\bibinfo{author}{Schmidt, P.}, \bibinfo{author}{Reiss, A.},
  \bibinfo{author}{D{\"{u}}richen, R.}, \bibinfo{author}{Marberger, C.},
  \bibinfo{author}{Laerhoven, K.V.}, \bibinfo{year}{2018}.
\newblock \bibinfo{title}{Introducing wesad, a multimodal dataset for wearable
  stress and affect detection}, in: \bibinfo{booktitle}{Proceedings of the 2018
  on International Conference on Multimodal Interaction, {ICMI} 2018, Boulder,
  CO, USA, October 16-20, 2018}, pp. \bibinfo{pages}{400--408}.
\newblock \URLprefix \url{https://doi.org/10.1145/3242969.3242985},
  \DOIprefix\doi{10.1145/3242969.3242985}.
\bibitem[{Sidirourgos et~al.(2011)Sidirourgos, Kersten and
  Boncz}]{DBLP:conf/cidr/SidirourgosKB11}
\bibinfo{author}{Sidirourgos, L.}, \bibinfo{author}{Kersten, M.L.},
  \bibinfo{author}{Boncz, P.A.}, \bibinfo{year}{2011}.
\newblock \bibinfo{title}{Sciborq: Scientific data management with bounds on
  runtime and quality}, in: \bibinfo{booktitle}{{CIDR} 2011, Fifth Biennial
  Conference on Innovative Data Systems Research, Asilomar, CA, USA, January
  9-12, 2011, Online Proceedings}, pp. \bibinfo{pages}{296--301}.
\newblock \URLprefix
  \url{http://cidrdb.org/cidr2011/Papers/CIDR11\_Paper39.pdf}.
\bibitem[{Thirumuruganathan et~al.(2019)Thirumuruganathan, Hasan, Koudas and
  Das}]{DBLP:journals/corr/abs-1903-10000}
\bibinfo{author}{Thirumuruganathan, S.}, \bibinfo{author}{Hasan, S.},
  \bibinfo{author}{Koudas, N.}, \bibinfo{author}{Das, G.},
  \bibinfo{year}{2019}.
\newblock \bibinfo{title}{Approximate query processing using deep generative
  models}.
\newblock \bibinfo{journal}{CoRR} \bibinfo{volume}{abs/1903.10000}.
\newblock \URLprefix \url{http://arxiv.org/abs/1903.10000},
  \href{http://arxiv.org/abs/1903.10000}{\tt arXiv:1903.10000}.
\bibitem[{Vieira et~al.(2011)Vieira, Razente, Barioni, Hadjieleftheriou,
  Srivastava, Jr. and Tsotras}]{DBLP:conf/icde/VieiraRBHSTT11}
\bibinfo{author}{Vieira, M.R.}, \bibinfo{author}{Razente, H.L.},
  \bibinfo{author}{Barioni, M.C.N.}, \bibinfo{author}{Hadjieleftheriou, M.},
  \bibinfo{author}{Srivastava, D.}, \bibinfo{author}{Jr., C.T.},
  \bibinfo{author}{Tsotras, V.J.}, \bibinfo{year}{2011}.
\newblock \bibinfo{title}{On query result diversification}, in:
  \bibinfo{booktitle}{Proceedings of the 27th International Conference on Data
  Engineering, {ICDE} 2011, April 11-16, 2011, Hannover, Germany}, pp.
  \bibinfo{pages}{1163--1174}.
\newblock \URLprefix \url{https://doi.org/10.1109/ICDE.2011.5767846},
  \DOIprefix\doi{10.1109/ICDE.2011.5767846}.

\end{thebibliography}

\end{document}